\documentclass[12pt,a4paper,reqno,final]{article}
\usepackage{graphicx}
\usepackage{amsmath}
\usepackage{authblk}
\usepackage[english]{babel}
\usepackage{mathtools}
\usepackage{esint}
\usepackage{lmodern}

\usepackage{graphicx}
\usepackage{amsmath}
\usepackage{authblk}
\usepackage{esint}
\usepackage{tikz}

\usepackage{enumitem}
\usepackage{upgreek}
\usepackage{mathtools}
\usepackage[numbers,sort&compress]{natbib}
\usepackage[margin=0.8in]{geometry}

\newcommand{\alast}{\alpha^{\ast}}
\usepackage{multirow}
\definecolor{refkey}{rgb}{0,0,1}
\definecolor{labelkey}{rgb}{0,0,1}
\usepackage{amsmath,amsthm,amssymb}
\usepackage{mathrsfs}
\numberwithin{equation}{section}
\usepackage{subcaption}
\usepackage[labelformat=simple]{subcaption}

\usepackage{enumitem}
\usepackage{mathtools}  
\mathtoolsset{showonlyrefs} 
\allowdisplaybreaks

\newtheorem{theorem}{Theorem}[section]

\newtheorem{remark}{Remark}[section]

\usepackage[normalem]{ulem}
\begin{document}

\title{Two--phase model of compressive stress induced on a surrounding hyperelastic medium by an expanding tumour}
\author[1]{Gopikrishnan C. Remesan\thanks{Email: gopikrishnan.chirappurathu@unimi.it}}
\author[2]{Jennifer A. Flegg\thanks{Email: jennifer.flegg@unimelb.edu.au}}
\author[3]{Helen M. Byrne\thanks{Email: helen.byrne@maths.ox.ac.uk}}

\affil[1]{\small{Dipartimento di Matematica ``F. Enriques'', Universit\`{a} degli Studi di Milano, 20133 Milan, Italy}}
\affil[2]{\small{University of Melbourne, Melbourne, Australia}}
\affil[3]{\small{Mathematical Institute, University of Oxford, Oxford, United Kingdom}}
\maketitle

\begin{abstract}
\emph{In vitro} experiments in which tumour cells are seeded in a gelatinous medium, or hydrogel,
show how mechanical interactions between tumour cells and the tissue in which they are embedded, together with
local levels of an externally-supplied, diffusible nutrient (e.g., oxygen), affect the tumour's growth dynamics. 
In this article, we present a mathematical model that describes these \emph{in vitro} experiments. 
We use the model to understand how tumour growth generates
mechanical deformations in the hydrogel and how these deformations in turn influence the tumour's growth.
The hydrogel is viewed as a nonlinear hyperelastic material and
the tumour is modelled as a two-phase mixture, comprising a viscous tumour cell phase and an isotropic, inviscid interstitial fluid phase. 
Using a combination of numerical and analytical techniques, we show how the tumour's growth dynamics change as the mechanical properties of the hydrogel vary.
When the hydrogel is soft, nutrient availability dominates the dynamics: the tumour evolves to a large equilibrium configuration where the proliferation
rate of nutrient-rich cells on the tumour boundary balances the death rate of nutrient-starved cells in the central, necrotic core. As
the hydrogel stiffness increases, mechanical resistance to growth increases and the tumour's equilibrium size decreases. Indeed,
for small tumours embedded in stiff hydrogels, the inhibitory force experienced by the tumour cells may be so large that the tumour is eliminated. 
Analysis of the model identifies parameter regimes in which the presence of the hydrogel drives tumour elimination. 

\vspace{0.5cm} 
\vspace{0.5cm} 
\end{abstract}

\section{Introduction}
During its earliest stage of avascular growth, a tumour consists of proliferating cells and lacks a vascular network. 
External nutrients that diffuse into the tumour provide the energy that the cells need to proliferate~\cite{hirschaeuser}. When the tumour reaches a size at which diffusive transport of nutrients can no longer sustain its energy 
requirements, the tumour acquires a blood supply from the surrounding tissue vasculature 
via a process known as angiogenesis~\cite{ferrara2002vegf}. 
Once the tumour has acquired its own network of blood vessels, its nutrient supply increases, enabling continued growth, expansion and the onset of malignancy~\cite{duffy2008}. For these reasons, it is important to understand the mechanisms that control the avascular tumour growth and, in particular, to identify strategies for its inihibition.

In addition to nutrient availability, many other phenomena are known to influence avascular tumour growth. 
These include interactions with immune cells and stromal cells~\cite{bissell2001putting,gonzalez2018} 
and mechanical effects \cite{chaudhuri2018mechanobiology, jain2014role}. 
Of particular interest in this paper is the impact that mechanical effects have on tumour growth; the \emph{in vitro} experimental work performed by Helmlinger \emph{et al.} providing
our primary motivation~\cite{Helmlinger1997}. By embedding small clusters of tumour cells in hydrogels of different stiffnesses, they showed how mechanical resistance can
inhibit a tumour's growth dynamics. 
As the tumour cells proliferate, the tumour increases in size and deforms the surrounding hydrogel.
The hydrogel, in turn, exerts a compressive stress on the tumour cells, which inhibits their net rate of cell proliferation and, thereby, reduces the equilibrium size to which
the tumour grows.  Subsequent studies have shown how externally imposed compressive stresses can affect cell proliferation and death rates
within multicellular tumour spheroids~\cite{rakesh1, delarue}.
Takao \emph{et al.}  showed that short periods of oscillatory compressive stress can also stimulate 
extensive cell death in breast cancer cell lines~\cite{Takaobio}. 
Mechanical stress drives other behavioural changes: 
it may activate tumour cells to become more motile and to model their tissue environment (e.g., the stiffness and fibrillar structure of the extracellular 
matrix~\cite{oncogenes}). Recent studies also indicate that mechanical stress has a critical influence in promoting metastasis and tumour progression~\cite{liu2020}. Taken together, these studies illustrate the 
need for increased understanding of the mechanisms by which mechanical stimuli inhibit tumour growth
and the clinical potential for exploiting tumour responses to mechanical stress to improve cancer treatment. 
Mathematical modelling represents a natural framework within which to address such questions. 
 
There is a vast literature devoted to mathematical models of the growth and response to treatment of solid tumours (see \cite{bull2022hallmarks, flegg2019mathematical, mathur2022optimizing, IMA::Byrne20061563, IMA::Roose2007179} for recent reviews). 
For example, a variety of mathematical frameworks have been used to study the growth of multicellular tumour spheroids, ranging from 
agent-based models~\cite{bull2020mathematical}, 
to time-dependent systems of ordinary differential equations~\cite{wallace2013properties}, and partial differential equations~\cite{gatenby1996}, including novel classes of free boundary problems~\cite{greenspan1972,Friedman:SIAM:2003, Friedman:2003, Friedman:2005,Friedman:2013}. The majority of these mathematical models have focussed on the impact that nutrient availability has on tumour growth and responses to 
treatment with radiotherapy and chemotherapy. However, a small number of models have studied the impact of mechanical stimuli on 
tumour growth~\cite{Byrne2003a, Ambrosi2004, Ambrosi2017}, including
several specialised to  understand the experimental results obtained by Helminger \emph{et al.}~\cite{Helmlinger1997,IMA::Chen2001,Roose:2003,Yan:2021}.
A common feature of the latter models is their treatment of the hydrogel as a hyperelastic material that restrains the tumour's growth.
The models differ in the constitutive assumptions that they use to model the tumour:
both Chen \emph{et al.}~\cite{IMA::Chen2001} and Roose \emph{et al.}~\cite{Roose:2003} view the tumour as a two-phase mixture
while Yan \emph{et al.}~\cite{Yan:2021} view it as a neo-Hookean elastic material. 
In~\cite{IMA::Chen2001} the tumour is viewed as a mixture of two inviscid fluids, cell proliferation and death are limited by local nutrient levels and necrosis is 
initiated when the pressures in the cell and fluid phases are equal. By contrast,
in~\cite{Roose:2003} the tumour is viewed as a poroelastic material with a solid, cellular phase and an inviscid fluid phase, cell proliferation and  death
are regulated by nutrient availability and also the local cell stress, and necrosis is neglected. 

In this paper, we introduce a new model to describe Helmlinger \emph{et al.}'s experiments, aiming to strike a balance
between the phenomenological approach used in~\cite{IMA::Chen2001} and the more detailed, poroelastic approach
used in~\cite{Roose:2003}. Following ~\cite{IMA::breward_2002}, we view the tumour as a two-phase mixture, 
consisting of a viscous tumour cell phase and an inviscid fluid phase.  
Our approach is flexible and builds naturally on the existing literature on multiphase models of solid tumour growth~\cite{Roose:2003,Ambrosi2017,Lemonetal,Byrne2003a}.
As such, it is readily extendable in that the assumptions that underpin both the tumour and 
hydrogel sub-models can be changed while maintaining the same general model framework. 

The remainder of the paper is organised as follows.  
The mathematical model is derived and cast in dimensionless form in Section~\ref{sec:mathematical_model}.  
Numerical solutions are presented in Section~\ref{sec:numerical_solutions}
In Section~\ref{sec:tumour_decay} we derive a simplified version of the model which we use to investigate how the mechanical properties of the hydrogel impact the tumour's growth dynamics and, in particular, to identify a critical value of the hydrogel stiffness above which tumour elimination is predicted and below which the tumour evolves to a non-trivial steady state. 
The paper concludes in Section~\ref{sec:discussion} where we discuss our findings and identify possible directions for future work.

 \section{Model development}
\label{sec:mathematical_model}
We view the tumour and hydrogel as two continuous materials, separated by a dynamic interface, referred to as the tumour boundary (see Figure~\ref{fig:tum_schema}).
In what follows,  we derive equations that describe how the tumour's size, composition,
and the position of the tumour boundary evolve over time and how these changes are coupled to the deformation of the hydrogel. For simplicity, we consider the 1D Cartesian geometry shown in Figure~\ref{fig:tum_schema}. %
The equations for the hydrogel and tumour are presented in
Subsections~\ref{sec:deform} and~\ref{sec:tumour_mod}, respectively, while the initial and boundary conditions are presented in Subsection~\ref{sec:initial_bouundary}. We non-dimensionalise the governing equations and discuss the parameter values used for numerical simulations in Subsection~\ref{sec:dless-gel}. For reference, lists of the dependent variables and model parameters that appear in the governing equations are presented in Tables~\ref{tab:mod_var} and~\ref{tab:mod_par}, respectively. 

\begin{figure}[h!]
	\centering
	\includegraphics[scale=0.9]{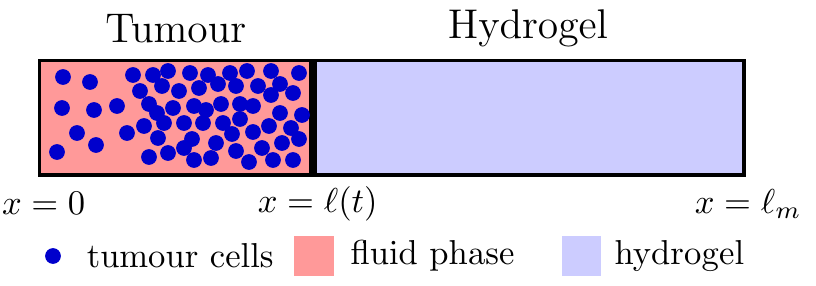}
	\caption{Schematic of the tumour and hydrogel. Here, $0  \leq x=\ell(t) \leq \ell_m$ denotes the position of the tumour-hydrogel interface at time $t$ and $\ell_m$ is the length of the domain in which the tumour and hydrogel are located.}
	\label{fig:tum_schema}
\end{figure}

\begin{table}[h!]
	\centering
	\begin{tabular}{||c|cc|c|c||}
		\hline
		\multirow{2}{*}{\textbf{}}       & \multicolumn{2}{c|}{\textbf{Parameter}}                         & \multirow{2}{*}{\textbf{Description}} & \multirow{2}{*}{\textbf{Dimension}} \\ \cline{2-3}
		& \multicolumn{1}{c|}{\textbf{Cell phase}} & \textbf{Fluid phase} &                                       &                                     \\ \hline \hline
		\multirow{5}{*}{\textbf{Tumour}} & \multicolumn{1}{c|}{$\alpha$}               & $\beta$                 & Volume fraction                       &    $1$                                 \\  
		& \multicolumn{1}{c|}{$u_{\alpha}$}            & $u_{\beta}$          & Velocity                              &    $\mathrm{cm}\cdot \mathrm{s}^{-1}$                                 \\ 
		& \multicolumn{1}{c|}{$p_{\alpha}$}            & $p_{\beta}$              & Pressure                              &       $\mathrm{g\cdot cm^{-1} \cdot s^{-2}}$                              \\
		& \multicolumn{1}{c|}{$\sigma_{\alpha}$}        & $\sigma_{\beta}$          & Stress                                &   $\mathrm{g\cdot cm^{-1} \cdot s^{-2}}$                                   \\ \cline{2-5} 
		& \multicolumn{2}{c|}{$c$}                                          & Nutrient concentration                &        $\mathrm{g\cdot cm^{-1} }$                              \\ \hline \hline
		\textbf{Hydrogel}                & \multicolumn{2}{c|}{$\sigma^{\rm H}$}                                   & Hydrogel stress                       &        $\mathrm{g\cdot cm^{-1} \cdot s^{-2}}$                              \\ \hline
	\end{tabular}
	\caption{Summary of the dependent variables used to describe the tumour and hydrogel.}
	\label{tab:mod_var} 
\end{table}

\begin{table}[h!]
	\centering
	\begin{tabular}{||c|c|c|c||}
		\hline
		& \multirow{1}{*}{\bf Parameter} & \multirow{1}{*}{\bf Description} & \multirow{1}{*}{\bf Dimension} \\ 
		\hline  \hline
		\multirow{4}{*}{\bf Tumour} 	& $\alpha^\ast$ & Cell packing density &  1 \\ 
		& $\mu_{\alpha}$ &Cell viscosity &  $\rm g\cdot cm^{-1}\cdot s^{-1}$ \\ 
		& $\gamma$ &Intracellular force coefficient  & $\rm g\cdot cm^{-1}\cdot s^{-2}$ \\  
		& {$k,k_1$} & Traction coefficients  & $\rm g\cdot cm^{-3}\cdot s^{-1}$\\ \hline \hline 
		\multirow{1}{*}{\bf Hydrogel} 
		& {$\vartheta,\,\nu$} & Compressibility parameters & $\rm g\cdot cm^{-1}\cdot s^{-2}$ \\ \hline 
	\end{tabular}
	\caption{Summary of the model parameters used to describe the mechanical properties of the tumour and hydrogel.}
	\label{tab:mod_par} 
\end{table}

\subsection{The hydrogel model}
\label{sec:deform}

The hydrogel surrounding the tumour is viewed as a hyperelastic material which deforms as the tumour expands. The mechanical stress that develops within the hydrogel acts on the tumour, inhibiting its growth. The dynamics of the tumour and hydrogel are coupled by imposing continuity of stress on the tumour boundary, $x = \ell(t)$ (see Subsection~\ref{sec:initial_bouundary}). 

In what follows, it is convenient to introduce 
$X \in (\ell(0),\ell_m)$ and $x \in (\ell(t),\ell_m)$ to represent the spatial coordinates within the undeformed (at $t=0)$ and deformed ($t > 0$) hydrogel. The \emph{deformation map} $\chi(t,X)$ maps points
$X \in (\ell(0),\ell_m)$ in the undeformed configuration to points $x \in (\ell(t),\ell_m)$ in the deformed configuration.

When describing the deformation of the hydrogel, we neglect inertial effects and assume that there are no external body forces. 
If we denote by 
$\sigma^{\rm H}$ the first Piola--Kirchoff stress tensor of the hydrogel, then the force balance on the hydrogel supplies: 
\begin{align} \label{eqn:gel_move}
\dfrac{\partial \sigma^{\rm H}}{\partial X} (t,X) = 0\;\; \text{ for} \; X \in (\ell(0),\ell_m) ;\;\text{ and }\; t > 0.
\end{align}
We denote by $G(t,X)$ the deformation gradient where 
$G(t,X) = \frac{\partial \chi}{\partial X}(t,X)$ for $X \in (\ell(0),\ell_m)$ and $t >0$. 
The first Piola-Kirchoff stress tensor is defined in terms of the strain energy density $\mathscr{W}_G$ via
$\sigma^{\rm H} = \partial \mathscr{W}_G/\partial G$, where $\mathscr{W}_G$ is the work done per unit volume in deforming the hydrogel. A wide range of strain energy density functionals have been proposed in the literature \cite{BERGSTROM2015209,YEOH1989425}. We follow Flory~\cite{flory1953, hong2008} and consider the following strain energy density:
\begin{equation}
\mathscr{W}_G =   \dfrac{\nu}{2} G^2 - \dfrac{3}{2}Nk_{\mathrm{B}}\mathfrak{T}_{\mathrm{abs}} -  \vartheta \log\left(G \right).
\label{eqn:flory}
\end{equation}
In Equation (\ref{eqn:flory}), $\nu$ and $\vartheta$ are nonnegative constants with the dimension of stress ($\rm g\cdot cm^{-1}\cdot s^{-2}$), $N$ is the ratio of the number of hydrogel polymer chains to the volume of the hydrogel in a dry state, $k_{\mathrm{B}}$ is the Boltzmann constant, and $\mathfrak{T}_{\mathrm{abs}}$ is the absolute temperature which is assumed constant.  
%
With $\sigma^\mathrm{H} := \frac{\partial  \mathscr{W}_G}{\partial G}$ and $\mathscr{W}_G$ defined by Equation~\eqref{eqn:flory}, it follows that
\begin{align}
\sigma^{\mathrm{H}} =   \nu\,G -  \vartheta G^{-1}.
\label{eqn:first_piola}
\end{align}

In what follows, for physically realistic solutions, we assume that $\vartheta \ge \nu$. Under this assumption, $\sigma^{\mathrm{H}} < 0$ and the hydrogel always exerts a compressive stress on the tumour (see Figure~\ref{fig:sighneg} where $\vartheta \ge \nu$ and $\sigma^{\mathrm{H}}<0$). If $\nu > \vartheta$, then $\sigma^{\mathrm{H}} > 0$ for $\ell(t) < \ell_m - (\ell_m - \ell_0)\sqrt{\vartheta/\nu} = \ell^\ast$ (see Figure~\ref{fig:sighpos}), and the hydrogel exerts a physically unrealistic, tensile stress on the tumour for $0 < \ell(t) < \ell^\ast$.  

\begin{figure}[h!]
	\centering
	\begin{subfigure}[c]{0.48\textwidth}
		\centering
		\caption{}
		\includegraphics[scale=0.8]{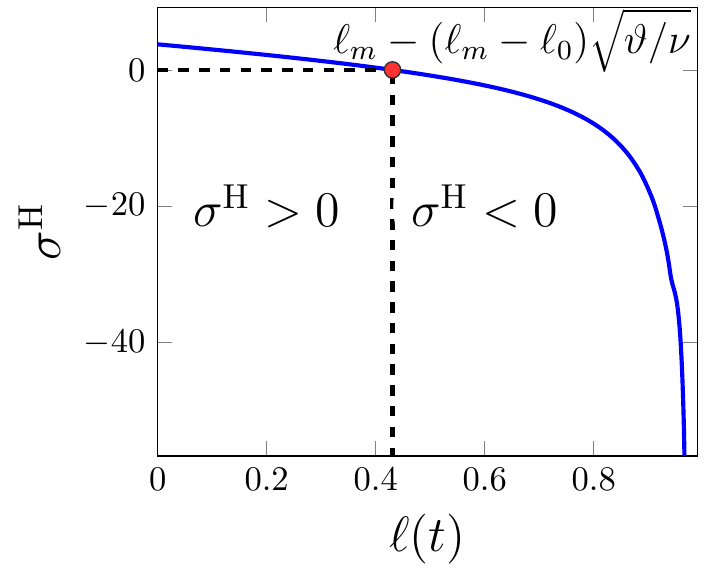}
		\label{fig:sighpos}
	\end{subfigure}
	\begin{subfigure}[c]{0.48\textwidth}
		\centering
		\caption{}
		\includegraphics[scale=0.8]{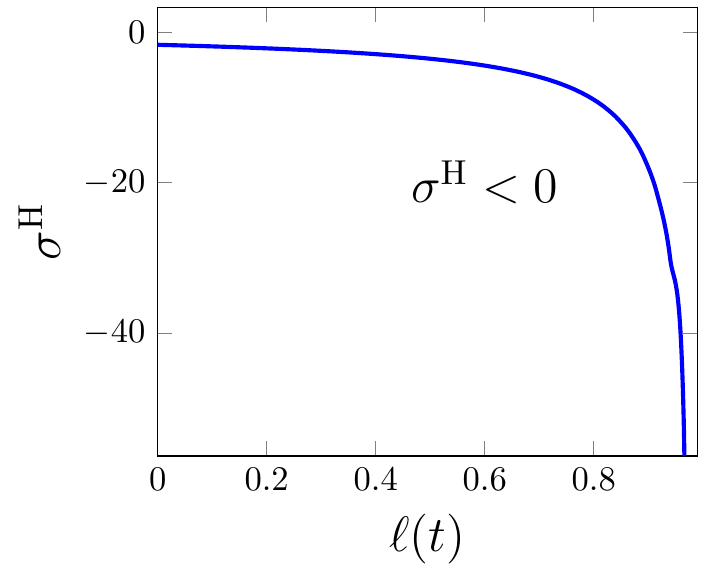}
		\label{fig:sighneg}
	\end{subfigure}
	\caption{Variation of $\sigma^{\mathrm{H}}$ with $\ell(t) \in (0,1)$. In (a), $\vartheta < \nu $ ($\vartheta = 2.0$, $\nu = 5.0$) and $\sigma^{\rm H}$ is positive (physically unrealistic) at values of $\ell(t)$ to the left of the vertical dotted line, on which 
		$\ell(t) = \ell_m - (\ell_m - \ell_0)\sqrt{\vartheta/\nu}$. In (b), $\vartheta > \nu$ ($\vartheta = 2.0$, $\nu = 0.1$) and $\sigma^{\rm H} < 0$ for all values of $\ell(t)$.}
\end{figure}

\noindent Using Equation~\eqref{eqn:first_piola} to substitute for $\sigma^\mathrm{H}$ in Equation~\eqref{eqn:gel_move}, we obtain
\begin{align}
\nu\dfrac{\partial^2 \chi}{\partial X^2} + \vartheta\left(\dfrac{\partial \chi}{\partial X} \right)^{-2} \dfrac{\partial^2 \chi}{\partial X^2} = 0.
\label{eqn:gel_kine}
\end{align}

\noindent Since $\frac{\partial \chi}{\partial X}$ is the ratio of two infinitesimal lengths and the linear ordering of the material in the hydrogel is preserved under any admissible transformations (compression and stretching), we impose the auxiliary condition $\frac{\partial \chi}{\partial X} > 0$. 
At each time $t >0$, the boundary points $X=\ell_0$ and $X=\ell_m$ are mapped to $x=\ell(t)$ and $x=\ell_m$. Therefore, we impose the following boundary conditions to close Equation~\eqref{eqn:gel_kine}:
\begin{align}
\chi(t,\ell_0) = \ell(t) \;\text{ and }\; \chi(t,1) = \ell_m.
\label{eqn:bdry_cds_gel}
\end{align}
With $\dfrac{\partial \chi}{\partial X} > 0$, Equations~\eqref{eqn:gel_kine} and ~\eqref{eqn:bdry_cds_gel} admit a unique, linear solution of the form: 
\begin{align}
\chi(t,X) = \left ( \frac{\ell_m - \ell(t)}{\ell_m - \ell_0} \right) X + \left ( \frac{\ell(t) - \ell_0}{\ell_m - \ell_0} \right).
\label{eqn:spatial_deform}
\end{align}
Recalling that $G=\frac{\partial \chi}{\partial X}(t,X)$, we now substitute  Equation~\eqref{eqn:spatial_deform} in Equation~\eqref{eqn:first_piola} to obtain the following expression for $\sigma^{\mathrm{H}}$, the stress in the hydrogel:
\begin{align}
\sigma^{\mathrm{H}}(t,X) =  \nu \left ( \dfrac{\ell_m - \ell(t)}{\ell_m - \ell_0}  \right ) - \vartheta \left(\dfrac{\ell_m - \ell_0}{\ell_m - \ell(t)} \right), \quad \mbox{for} \;\; X \in (\ell(0),\ell_m).
\label{eqn:gel_stress1}
\end{align}

\begin{remark}[Spatial variation of the stress in the hydrogel] 
	From Equation~\eqref{eqn:gel_stress1}, it is clear that the stress $\sigma^{\mathrm{H}}$ depends implicitly on time, via the tumour length $\ell(t)$. 
	The stress $\sigma^{\mathrm{H}}$ is independent of spatial position $X$ because 
	the deformation in the hydrogel depends linearly on $X$ (see Equation~\eqref{eqn:spatial_deform}).
	In Sections~\ref{sec:numerical_solutions} and~\ref{sec:tumour_decay} we show that when the stress in the hydrogel is defined by
	Equation~\eqref{eqn:gel_stress1}, the analysis and numerical solution of
	the model simplify greatly. 
\end{remark}

\subsection{The two-phase model of tumour growth}
\label{sec:tumour_mod}

The tumour is viewed as a two-phase mixture of viscous tumour cells and inviscid fluid, embedded within a deformable hydrogel that acts as an external source of vital nutrients,
here taken to be oxygen. 
As indicated in Table~\ref{tab:mod_var}, we associate with the tumour cell and fluid phases, volume fractions $\alpha$ and $\beta$ (dimensionless),  velocities $u_\alpha$ and $u_\beta$ ($\rm cm\cdot s^{-1}$), pressures $p_\alpha$ and $p_\beta$, and stress tensors $\sigma_\alpha$ and $\sigma_\beta$ ($\rm g\cdot cm^{-1} \cdot s^{-2}$). We note that, for 1D Cartesian geometry considered here, the dependent variables are scalars which depend on spatial position $0 \leq x \leq \ell(t)$ and time $t > 0$.  

Conservation of mass applied to the tumour cell and fluid phases supplies
\begin{align}
\dfrac{\partial \alpha}{\partial t} + \dfrac{\partial }{\partial x}(u_\alpha \alpha) = \mathscr{F}_\alpha(\alpha,c)\;\;\text{ and }\;\;\dfrac{\partial \beta}{\partial t} + \dfrac{\partial}{\partial x}(u_\beta \beta) = -\mathscr{F}_\alpha(\alpha,c),
\label{eqn:vf}
\end{align}
where $\mathscr{F}_\alpha$ and $\mathscr{F}_\beta$ are the net production rates of tumour cells and fluid respectively and $c$ ($\rm g\cdot cm^{-1}\cdot s^{-1}$)  is the concentration of an externally supplied nutrient, here taken to be oxygen. We assume that there are no external sources or sinks of mass within the tumour, so that mass is converted from one phase to the other ({\em i.e.,} $\mathscr{F}_\alpha = - \mathscr{F}_\beta$). We assume further that there are no voids within the tumour so that
\begin{equation}
\alpha + \beta = 1 \quad \mbox{for} \; 0 \leq x \leq \ell(t). 
\label{eqn:no_voids}
\end{equation}
Following \cite{IMA::breward_2002}, we implement the following functional form for $\mathscr{F}_\alpha$:
\begin{equation}
\mathscr{F}_\alpha(\alpha,c) = \underbrace{\frac{S_0 c}{1 + S_1c)}}_{\equiv \, b(c)} \:  \alpha (1 - \alpha) -  \underbrace{\frac{(S_2 + S_3 c)}{(1 + S_4c)}}_{\equiv \, d(c)} \: \alpha,    
\end{equation}
where $b(c)$ and $d(c)$ represent the birth and death rates of the tumour
cells, and $S_0\,({\rm g^{-1}\cdot cm})$, $S_1\,({\rm g^{-1}\cdot cm\cdot s})$, $S_2 \,({\rm s^{-1}})$, $S_3\,({\rm g^{-1}\cdot cm})$ and $S_4 \,({\rm g^{-1}\cdot cm\cdot s})$ are positive parameters whose default values are presented in Table~\ref{tab:model_param}.
We assume that the oxygen concentration satisfies a reaction-diffusion equation of the form
\begin{align}
\label{eqn:ot}
\dfrac{\partial c}{\partial t} = \eta \dfrac{\partial^2 c}{\partial x^2}  -\dfrac{Q_0\alpha c}{1 + Q_1 c},
\end{align}
where the positive constant $\eta$ represents the oxygen diffusion coefficient, and the positive constants $Q_0$ and $Q_1$ describe how the rate at which tumour cells consume oxygen increases as oxygen levels increase. Since the time--scale of tumour cell proliferation is typically much larger that that of oxygen diffusion, henceforth, we make a quasi-steady state approximation, wherein $\partial c/\partial t \approx 0$~\cite{ward_1,IMA::breward_2002,Byrne_prezziozi_2003,Byrne2003a} and   Equation~\eqref{eqn:ot} reduces to give
\begin{align}
\label{eqn:ot_QS}
0 = \eta \dfrac{\partial^2 c}{\partial x^2} - \dfrac{Q_0\alpha c}{1 + Q_1 c}.
\end{align}

\medskip
Applying momentum balances to the tumour cell and fluid phases, and neglecting inertial effects, we have
\begin{align}
\dfrac{\partial (\alpha \sigma_{\alpha})}{\partial x} + F_{\alpha\beta} = 0\;\;\text{ and }\;\;  \dfrac{\partial (\beta \sigma_{\beta})}{\partial x} + F_{\beta\alpha} = 0, 
\label{eqn:bal_momen}
\end{align}
where $F_{\alpha\beta}$ and $F_{\beta\alpha}$ represent the forces exerted by the fluid phase on the cells, and vice versa and, in the absence of external forces, $F_{\alpha\beta} = -F_{\beta\alpha}$. We assume that 
$F_{\alpha\beta}$ comprises two terms: a drag term due to relative motion of the two phases, of the form $k_1 \alpha \beta (u_\beta - u_\alpha)$,
where $k_1$ ($\rm g\cdot cm^{-3}\cdot s^{-1}$) is a positive constant; and, an interfacial force of the form $p_{\beta} \partial \alpha/ \partial x$. 
Combining these assumptions, we have that
\begin{equation}  \label{eqn:Fab}
F_{\alpha\beta} \equiv - F_{\beta \alpha} = k_1 \alpha \beta (u_\beta - u_\alpha) + p_{\beta} \frac{\partial \alpha}{ \partial x}.
\end{equation}

Since the fluid phase is inviscid and the cell phase is viscous, the stresses $\sigma_\alpha$ and $\sigma_\beta$ are 
\begin{align}
\sigma_{\alpha} = -p_{\alpha}+ 2\mu_{\alpha} \dfrac{\partial u_{\alpha}}{\partial x} \;\;
\text{ and }\;\;
\sigma_{\beta} = -p_{\beta},
\label{eqn:stress_tumour}
\end{align}
where $\mu_{\alpha}$ ($\rm g\cdot cm^{-1}\cdot s^{-1}$) is the coefficient of viscosity.

Adding the two equations in~\eqref{eqn:bal_momen}, and substituting for $\sigma_{\alpha}$ and $\sigma_{\beta}$ from~\eqref{eqn:stress_tumour}, it follows that
\begin{align}
2\mu_{\alpha} \dfrac{\partial }{\partial x} \left(\alpha \dfrac{\partial u_{\alpha}}{\partial x}\right) = \dfrac{\partial }{\partial x}\left(\alpha p_{\alpha} + \beta p_{\beta} \right).
\label{eqn:vel_simp1}
\end{align}
Following \cite{IMA::breward_2001}, we view the cells as bags of water, with an additional pressure to account for cell--cell interactions. Under these assumptions, we have that
\begin{equation} \label{eqn:palpha}
p_\alpha = p_{\beta} + 
\underbrace{\gamma\dfrac{(\alpha - \alast)^{+}}{(1 - \alpha)^2}}_{\equiv \, \Sigma(\alpha)},
\end{equation}
where, $\Sigma(\alpha)$ ($\rm g\cdot cm^{-1} \cdot s^{-2}$) quantifies the pressure due to cell--cell interactions, and $\gamma$ ($\rm g\cdot cm^{-1} \cdot s^{-2}$)  and $0 < \alast < 1$ (dimensionless) are positive constants and $x^{+} := \max(x,0)$ for $x \in \mathbb{R}$. If $0 < \alpha < \alpha^{\ast}$, then the cells are too sparse to experience any interactions and $\Sigma(\alpha) = 0$; if $\alpha^{\ast} < \alpha < 1$, then the cells are densely packed and repel each other, and $\Sigma(\alpha) > 0$. Additional mechanisms, including chemotaxis~\cite{ByrneOwen,Lemonetal}, can be incorporated into the functional form of $\Sigma$. Here, for simplicity, we focus on cell-cell interactions.  

If we substitute $p_\alpha = p_{\beta} + \Sigma(\alpha)$	in Equation~\eqref{eqn:vel_simp1}, and use the no voids relation ($\alpha + \beta  = 1$), then we obtain the following partial differential equation for $u_\alpha$:
\begin{align}
2\mu_{\alpha} \dfrac{\partial }{\partial x} \left(\alpha \dfrac{\partial u_{\alpha}}{\partial x}\right) = \dfrac{\partial p_{\beta}}{\partial x} + \dfrac{\partial}{\partial x} (\alpha \Sigma(\alpha)).
\label{eqn:vel_simp2}
\end{align}
In order to arrive at our final model, it remains to derive expressions for $p_\beta$ and $u_\beta$ in terms of $\alpha$ and $u_\alpha$. Adding Equations~\eqref{eqn:vf}, we deduce that $\frac{\partial}{\partial x} (\alpha u_{\alpha} + \beta u_{\beta}) = 0$. 
Integrating this identity with respect to $x$, and assuming that the tumour is symmetric about the tumour centre, $x=0$, so that $u_{\alpha}(0,t) = 0 = u_{\beta}(0,t)$, we have that 
\begin{equation}
u_{\beta} = -\dfrac{\alpha u_{\alpha}}{(1 - \alpha)}.
\label{eqn:ubeta}
\end{equation}
Using Equations~\eqref{eqn:Fab} and \eqref{eqn:ubeta} to substitute for $F_{\beta\alpha}$ and $u_\beta$, and noting that $\sigma_{\beta} = -p_{\beta}$,   Equation~\eqref{eqn:bal_momen} reduces to give the following expression for $p_\beta$: 
\begin{align}
\dfrac{\partial p_{\beta}}{\partial x} = \dfrac{k_1 \alpha}{1 - \alpha} u_{\alpha}.
\label{eqn:p_beta}
\end{align}
We substitute from Equation~\eqref{eqn:p_beta} into Equation~\eqref{eqn:vel_simp2} to obtain the following partial differential equation for $u_\alpha$:
\begin{align}
2\mu_{\alpha} \dfrac{\partial }{\partial x} \left(\alpha \dfrac{\partial u_{\alpha}}{\partial x}\right) = \dfrac{k_1 \alpha}{1 - \alpha} u_{\alpha} + \dfrac{\partial}{\partial x} (\alpha \Sigma(\alpha)).
\label{eqn:vel_simp3}
\end{align}
The tumour boundary $x = \ell(t)$ marks the interface between the tumour and the hydrogel.  We assume that it moves with the local cell velocity there so that
\begin{equation}
\frac{\mathrm{d} \ell (t)}{\mathrm{d}t} = u_{\alpha}(t,\ell(t)).
\label{bdrequation}
\end{equation}

\noindent \textbf{Reduced model:} Our two-phase model of tumour growth reduces to three partial differential equations for the cell volume fraction ($\alpha$), cell velocity ($u_{\alpha}$), nutrient concentration ($c$), and an ordinary differential equation for the tumour length ($\ell(t))$: 
\begin{align} \label{eq.td1}
\dfrac{\partial \alpha}{\partial t} + \dfrac{\partial }{\partial x}(u_\alpha \alpha) ={}& b(c)  \alpha (1 - \alpha)- d(c) \alpha,\\  \label{eq.td2}
2\mu_{\alpha} \dfrac{\partial }{\partial x} \left(\alpha \dfrac{\partial u_{\alpha}}{\partial x}\right) ={}& \dfrac{k_1 \alpha}{1 - \alpha} u_{\alpha} + \dfrac{\partial}{\partial x} (\alpha \Sigma(\alpha)), \\  \label{eq.td3}
\eta \dfrac{\partial^2 c}{\partial x^2}  ={}& \dfrac{Q_0\alpha c}{1 + Q_1 c},\;\;	\text{ and } \\
\frac{\mathrm{d} \ell(t)}{\mathrm{d}t} ={}& u_{\alpha}(t,\ell(t)).
\end{align}
The variables $\beta$ and $u_{\beta}$ are obtained from the relations $\alpha + \beta  = 1$ and $u_{\beta} = -\alpha u_{\alpha}/(1 - \alpha)$, while $p_\alpha$ 
and $p_\beta$ are defined by
Equations~\eqref{eqn:palpha} and~\eqref{eqn:p_beta}, and $\sigma_\alpha$ and $\sigma_\beta$ are defined by
Equation~\eqref{eqn:stress_tumour}.

\subsection{Initial and boundary conditions}
\label{sec:initial_bouundary}
The model is closed by imposing appropriate initial and boundary conditions. 
The initial tumour length and volume fraction are
\begin{equation} \label{eqn:ini}
\ell(0) = \ell_{\rm in} \;\text{ and }\;\alpha(0,x) = \alpha_{\rm in}(x) \quad \mbox{for} \;\; x \in (0,\ell_{\rm in}),
\end{equation}
where $0 < \ell_{\rm in} < \ell_{\rm m}$ and $0 < \alpha_{\rm in}(x) < 1$. 
To ensure symmetry about the tumour centre, we impose the following boundary conditions:
\begin{equation}
u_{\alpha}(t,0) = 0 \;\; \mbox{and} \;\;  \frac{\partial c}{\partial x}(t,0) = 0.
\label{eqn:bc_x0}
\end{equation}
We assume that the fluid flows freely across the tumour  boundary and, hence,  fix $p_{\beta} (t, \ell(t)) = 0$. 
With $\sigma^{\mathrm{H}}$ defined by Equation~\eqref{eqn:gel_stress1} and
$p_{\beta} (t, \ell(t)) = 0$, continuity of stress at the tumour--hydrogel interface then supplies, 
\begin{align} \label{eqn:stress_cont}
\sigma_{\alpha} ( t, \ell(t)) = 2 \mu_{\alpha} \dfrac{\partial u_{\alpha}}{\partial x}(t,\ell(t)) - \Sigma({\alpha}(t,\ell(t))) = \sigma^{\mathrm{H}}(t, \ell(t))
\;\;\text{ for }\;\; t > 0.
\end{align}
We suppose that the nutrient concentration is maintained at a constant value of $c_{\rm out}  > 0$ on the tumour boundary so that 
\begin{equation} \label{eqn:bc_nutrlt}
c(t,\ell(t)) = c_{\rm out}\;\; \text{ for } t > 0. 
\end{equation}

\begin{figure}[htp]
	\centering
\includegraphics[scale=0.7]{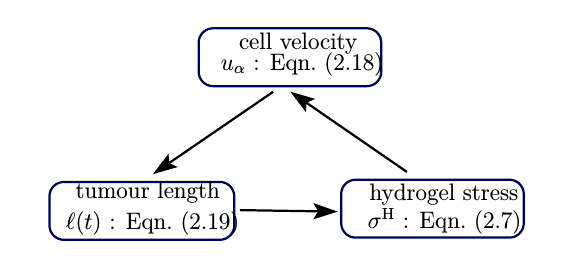}
	\caption{Schematic diagram showing how the stress in the hydrogel, $\sigma^{\mathrm{H}}(t)$, the tumour cell velocity, $u_{\alpha}(t,\ell(t))$, and the tumour size, $\ell(t)$, are coupled. }
	\label{fig:coupling}
\end{figure}

\noindent {\em Model summary.} The  tumour dynamics are governed by Equations~\eqref{eq.td1},~\eqref{eq.td2}, and~\eqref{eq.td3}. The initial conditions are specified by Equation~\eqref{eqn:ini} and the boundary conditions are given by Equations~\eqref{eqn:bc_x0},~\eqref{eqn:stress_cont}, 
and~\eqref{eqn:bc_nutrlt}. The boundary condition \eqref{eqn:stress_cont} couples the tumour dynamics to $\sigma^{\mathrm{H}}$, the stress in the hydrogel, which is defined via Equation~\eqref{eqn:gel_stress1}.

\subsection{Dimensionless model}
\label{sec:dless-gel}

Before proceeding, it is convenient to recast our model in dimensionless form. We rescale time $t$ with $t_{\mathrm{dim}}$, a typical timescale for tumour cell proliferation; here we fix $t_{\mathrm{dim}} = (1 + S_1c_{\mathrm{out}})/S_0c_{\mathrm{out}}$, which is the timescale for cell proliferation when the nutrient concentration is at its maximum value. 
We map the spatial domain onto the unit interval by rescaling the spatial coordinate, $x$, and the position of the tumour-hyrdogel interface, $\ell(t)$, with $\ell_{\mathrm{dim}} = \ell_m$, the size of the domain. For completeness, in Table~\ref{tab:dless_quan} we state the dimensionless variables induced by this rescaling, using primes to denote dimensionless quantities. 

\begin{table}[htp]
	\centering
	\begin{tabular}{||c||c||}
		\hline 
		{\bf Dimensional quantities}	& {\bf Dimensionless quantities} \\ \hline \hline
		Space, $x$, and time, $t$ & $x' = x/\ell_{\mathrm{dim}},\;\;t' = t/t_{\mathrm{dim}}$ \\ \hline 
		Cell volume fraction, $\alpha$ & $\alpha' = \alpha$ \\ \hline 
		Cell velocity, $u_\alpha$ & $\;u_{\alpha}' = \dfrac{t_{\mathrm{dim}}}{\ell_{\mathrm{d}}}u_\alpha$. \\ \hline
		Nutrient concentration, $c$ & $c' = c/c_{\mathrm{out}}$ \\ \hline 
		Tumour length, $\ell$ & $\ell' = \ell/\ell_{\mathrm{dim}}$  
		\\ \hline
		Cell and fluid stresses, $\sigma_{\alpha}$
		and $\sigma_\beta$ & $\sigma_{\alpha}'= \sigma_\alpha/\gamma$ and $\sigma_{\beta}'= \sigma_\beta/\gamma$\\ \hline
		Cell and fluid pressures, $p_\alpha$ and $p_\beta$& $p_{\alpha}'= p_\alpha/\gamma$ and $p_{\beta}'= p_\beta/\gamma$\\ \hline 
		Hydrogel stress, $\sigma^{\mathrm{H}}$ & $\sigma^{\mathrm{H}'}= \sigma^{\mathrm{H}}/\gamma$ \\ \hline
	\end{tabular}
	\caption{Summary of the dimensionless model variables.}
	\label{tab:dless_quan}
\end{table}

Under this rescaling, the equations governing the tumour-hydrogel dynamics transform as follows:

\begin{subequations} \label{cont_model}
	\begin{align}
	\label{eqn:vf_dless}
	\dfrac{\partial \alpha'}{\partial t'} + \dfrac{\partial }{\partial x'}(u_{\alpha}' \alpha') ={}& \alpha' 
	(1 - \alpha') \dfrac{(1 + s_1')c'}{1 + s_1'c'} - \alpha' \dfrac{s_2' + s_3'c'}{1+ s_4'c'}, \\
	\label{eqn:cv_dless}
	- \mu \dfrac{\partial }{\partial x'} \left(\alpha' \dfrac{\partial u_{\alpha}'}{\partial x'}\right) + \dfrac{k \alpha' u'}{1 - \alpha'}  ={}&  -\dfrac{\partial}{\partial x'} \left(\alpha' \dfrac{(\alpha' - \alast)^{+}}{(1 - \alpha')^2} \right), \\
	\eta'  \dfrac{\partial^2 c'}{\partial x'^2} ={}& \dfrac{Q\alpha' c'}{1 + \widehat{Q}_1 c'}, \;\;\text{ and }
	\label{eqn:ot_dless} \\
	\label{eqn:bvel_dless}
	\dfrac{\mathrm{d}\ell'}{\mathrm{d}t'} ={}& u_{\alpha'}(t',\ell'(t')), 
	\end{align}
	\noeqref{eqn:bvel_dless} subject to the initial and boundary conditions
	\begin{gather}
	{\alpha}'(0,x') = \alpha_{\rm in}'(x') \;\forall x' \in (0,\ell_{\rm in}'),\;\;{\ell}'(0) = \ell_{\rm in}',
	\label{eqn:in_cond}
	\end{gather}
	\begin{gather}
	\label{eqn:bdr_cond_1}
	{u}_{\alpha}'(t',0) = 0,\;\dfrac{\partial {c}'}{\partial x'}(t',0) = 0, \\ 	\label{eqn:bdr_cond_2}
	\mu \dfrac{\partial u_{\alpha}'}{\partial x'}(t',{\ell'}(t')) - \dfrac{(\alpha'(t',\ell'(t')) - \alast)^{+}}{(1 - \alpha'(t',\ell'(t')))^2} = \sigma^{\mathrm{H}}\; \text{ and }\;{c'}(t',{\ell'}(t')) = 1,  	\end{gather}
	\label{eqn:model}
	\noeqref{eqn:cv_dless,eqn:stress_hydrogel,eqn:in_cond,eqn:bdr_cond_1,eqn:bdr_cond_2}
	
	\noindent where 
	
	\begin{equation}
	\label{eqn:stress_hydrogel}
	\sigma^{\mathrm{H}} = \nu \left(\dfrac{1 - \ell'(t')}{1 - \ell_0'} \right) - \vartheta \left(\dfrac{1 - \ell_0'}{1 - \ell'(t')} \right).\\
	\end{equation}	
\end{subequations}

\noindent In Equations~(\ref{eqn:model}), we have introduced the following dimensionless parameter groupings:  
\begin{gather}
s_1' = S_1 c_{\mathrm{out}},\;\;s_2' = t_{\mathrm{dim}} S_2,\;\; s_3' = t_{\mathrm{dim}}S_3/c_{\mathrm{out}},\;\;
s_4' = S_4 c_{\mathrm{out}}, \\
\;\;{ k = \frac{k_1\ell_{\mathrm{dim}}^2}{\gamma t_{\mathrm{dim}}}},\;\;\mu = \frac{2\mu_{\alpha}}{\gamma t_{\mathrm{dim}}}, \\
Q = Q_0t_{\mathrm{dim}},\;\;\widehat{Q}_1 = Q_1 c_{\mathrm{out}},\;\;\eta' = t_{\mathrm{dim}}\eta/\ell_{\mathrm{dim}}^2, \\
\nu' = \nu/\gamma,\;\;\vartheta'= \vartheta/\gamma, \ell_{m}' = \ell_m/\ell_{\mathrm{dim}} \text{ and }\ell_{\rm in}' = \ell_{\rm in}/\ell_{\mathrm{dim}}.
\end{gather}
For notational ease, the prime symbols are dropped in the sequel. 

\subsection{Numerical Method}
Equations~\eqref{eqn:vf_dless}--\eqref{eqn:ot_dless} are defined on a time--dependent domain, with $0 \leq x \leq \ell(t)$. In order to construct numerical solutions to the governing equations, we use the scaling $\xi = x/\ell(t)$ to map the time-dependent domain $[0,\ell(t)]$ to the unit interval, $[0,1]$.    Similar coordinate transformations have been used by other authors to solve similar free boundary problems numerically (see~\cite{ward_1,yao} and the references therein).
Under this scaling, the model equations become
\begin{subequations}
	\label{eqn:fixed_prob}
	\begin{align}
	\label{eqn:df_cvf}
	\dfrac{\partial \alpha}{\partial t} - \dfrac{\xi}{\ell} \dfrac{\mathrm{d}\ell}{\mathrm{d} t} \dfrac{\partial \alpha}{\partial \xi}  + \dfrac{1}{\ell} \dfrac{\partial (u_{\alpha} \alpha)}{\partial \xi} ={}&  \alpha (1 - \alpha) \dfrac{(1 + s_1) c}{1 + s_1 c} - \dfrac{s_2 + s_3c}{1 + s_4c} \alpha, \\
	\label{eqn:df_cv}
	- \mu \dfrac{\partial }{\partial \xi} \left( \alpha \dfrac{\partial u_{\alpha}}{\partial \xi}\right) + \dfrac{\ell^2 k \alpha u_{\alpha}}{1 - \alpha} ={}& -\ell \dfrac{\partial}{\partial \xi}  \left ( \dfrac{(\alpha - \alpha^\ast)^{+}}{(1 - \alpha)^2} \right ),  \\
	\label{eqn:df_nc}
	\eta \dfrac{\partial^2 c }{\partial \xi^2}  ={}& \dfrac{Q \ell^2 \alpha c}{1 + \widehat{Q}_1 c}, \text{ and } \\
	\label{eqn:df-bdr}
	\dfrac{\mathrm{d} \ell}{\mathrm{d} t}  ={}& u(t,1),
	\end{align}
	\noeqref{eqn:df-bdr,eqn:df_cv}
	and the initial and boundary conditions are transformed to give
	\begin{gather}
	{\alpha}(0,\xi) = \alpha_{\rm in}(\xi)\;\forall \xi \in (0,1),\;\;{\ell}(0) = \ell_{\rm in},
	\label{eqn:transin_cond} \\
	\label{eqn:transbdr_cond_1}
	{u_{\alpha}}(t,0) = 0,\; \dfrac{\partial {c}}{\partial \xi}(t,0) = 0,\\
	\mu \dfrac{\partial {u_{\alpha}}}{\partial \xi}(t,1) = \ell(t) \left(\dfrac{(\alpha(t,1) - \alast)^{+}}{(1 - \alpha(t,1))^2}+ \sigma^{\mathrm{H}} \right),\; \text{ and }\;{c}(t,1) = 1 \quad \forall t \in (0,T), \label{eqn:transbdr_cond_2}
	\end{gather} \noeqref{eqn:transin_cond,eqn:transbdr_cond_1,eqn:transbdr_cond_2}
	where
	\begin{align}
	\label{eqn:gel}
	\sigma^{\rm H}= \nu \left(\dfrac{1 - \ell(t)}{1 - \ell_{\rm in}} \right) - \vartheta \left( \dfrac{1 - \ell_{\rm m}} {1 - \ell(t)}\right).
	\end{align}
\end{subequations}
\noindent 
We employ an upwind finite volume method to solve the hyperbolic  partial differential equation 
for the cell volume fraction, $\alpha$,  Equation~\eqref{eqn:df_cvf}; 
finite volume methods are appropriate here as they conserve mass at the discrete level~\cite{eymard}. 
Lagrange $P_1$ finite element methods~\cite{brenner2008mathematical,alexander} are used to solve the elliptic equations for the 
cell velocity, $u_\alpha$, and the nutrient, $c$,  Equation~\eqref{eqn:df_cvf} and Equation~\eqref{eqn:df_nc}. 
A backward Euler method is used to integrate the time-dependent ordinary differential for the tumour length, $\ell(t)$, 
Equation~\eqref{eqn:df-bdr}.  Further details of the numerical discretisation are omitted here. 
For formal definitions of the numerical methods and further justification of the above choices, 
we refer the interested reader to~\cite{DNR19,remesan_1}.

\subsection{Parameter Values}
In the absence of suitable experimental data with which to estimate the model parameters, values have been chosen to illustrate the range of behaviours that the model exhibits. Unless otherwise stated, we use the dimensionless parameter values in Table~\ref{tab:model_param} to generate numerical simulations of the governing equations.

\begin{table}[htp]
	\centering
	\begin{tabular}{||c|c||c | c||}
		\hline 
		{\bf Parameter} & {\bf Value} & {\bf Parameter} & {\bf Value}\\ \hline \hline
		$k$ & 1600 & $\mu$ & 1 \\ \hline
		$Q$ & 0.5 & $\widehat{Q}_1$ & 0 \\ \hline
		$s_1,\,s_4$ & 10 & 	$s_2,\,s_3$ & 0.5 \\ \hline
		$\alast$ & 0.8 & $\eta$ & 1/1600 \\ \hline
		$\alpha_0$ & 0.8 & $c_0$ & 1 \\ \hline
		$\ell_{\rm in}$ & 1/40 & - & - \\ \hline
	\end{tabular}
	\caption{Summary of the default values of the dimensionless parameters used to generate numerical simulations of Equations~\eqref{eqn:fixed_prob}.}
	\label{tab:model_param}
\end{table}

\section{Numerical results}
\label{sec:numerical_solutions}
In this section we present numerical simulations which show how the tumour's growth dynamics change as the stiffness of the hydrogel increases.
First, in Subsection~\ref{sec:free_sus}, we consider tumour growth in free suspension, wherein the effects of the hydrogel are absent. 
In Subsection~\ref{sec:small_tum}, we show how embedding the tumour in a deformable hydrogel influences its growth dynamics.

\subsection{Tumour growth in free suspension}
\label{sec:free_sus}
We simulate growth in free suspension by setting $\vartheta = 0 = \nu$ in Equation~\eqref{eqn:fixed_prob}.
In this case, $\sigma^{\mathrm{H}} = 0$ so that the tumour does not experience any mechanical resistance to its growth; 
its growth is limited by nutrient availability and the domain size (scaled to unity for the dimensionless model). 
Unless otherwise stated, all numerical results are generated using 
the dimensionless parameter values stated in Table~\ref{tab:model_param}, with simulations performed for $0 < t \leq T = 100$. 
Plots showing the time evolution of the nutrient concentration, cell volume fraction, cell velocity, and cell stress are presented
in Figure~\ref{fig:fres_sus_gr} while the corresponding time evolution of the outer tumour boundary, $x=\ell(t)$ is presented in Figure~\ref{fig:rad_time_theta}.

\begin{figure}[htp]
	\begin{minipage}{0.99\textwidth}
		\resizebox{.95\textwidth}{!}{	
			\centering
			\begin{subfigure}[t]{0.45\textwidth}
				\centering
				\caption{}
				\includegraphics[scale=0.6]{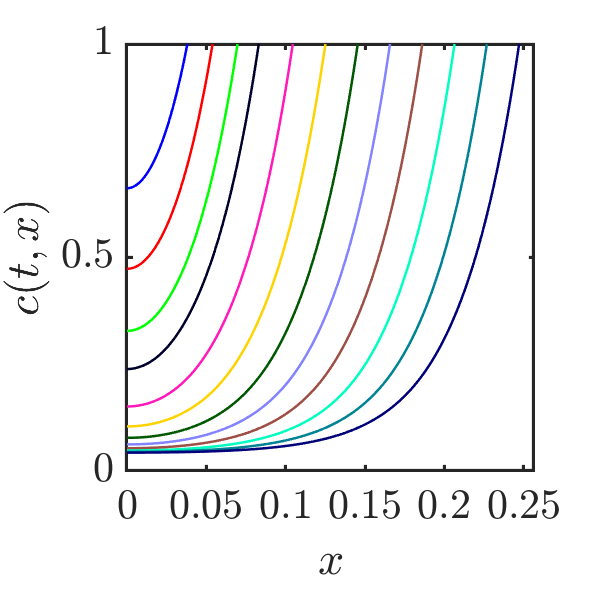}
				\label{fig:nutr00}
			\end{subfigure}
			\begin{subfigure}[t]{0.45\textwidth}
				\centering
				\caption{}
				\includegraphics[scale=0.6]{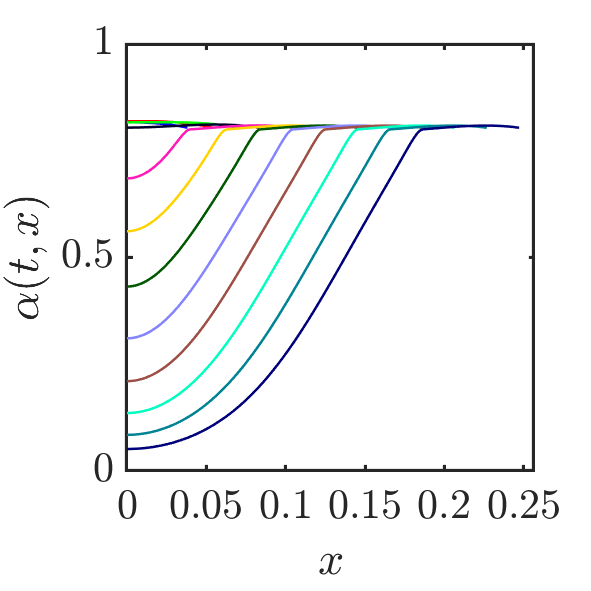}
				\label{fig:vf00}
			\end{subfigure} 
		} 
	\end{minipage}
	\begin{minipage}{0.99\textwidth}
		\resizebox{.95\textwidth}{!}{	
			\begin{subfigure}[t]{0.45\textwidth}
				\centering
				\caption{}
				\includegraphics[scale=0.6]{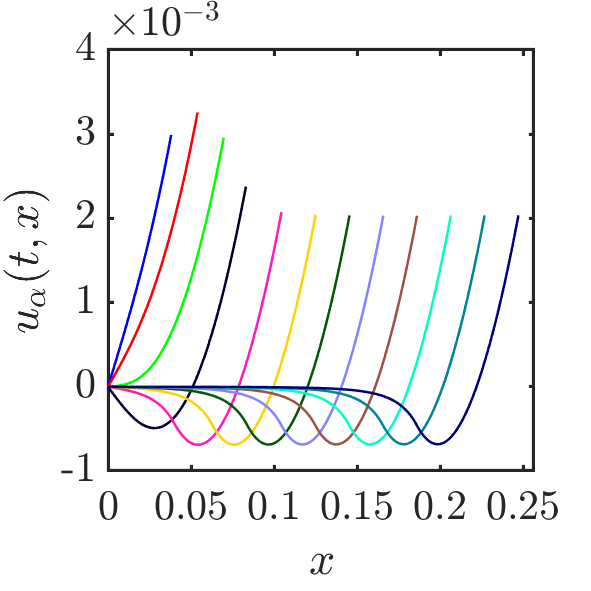}
				\label{fig:fs_cn2}
			\end{subfigure}
			\begin{subfigure}[t]{0.45\textwidth}
				\centering
				\caption{}
				\includegraphics[scale=0.6]{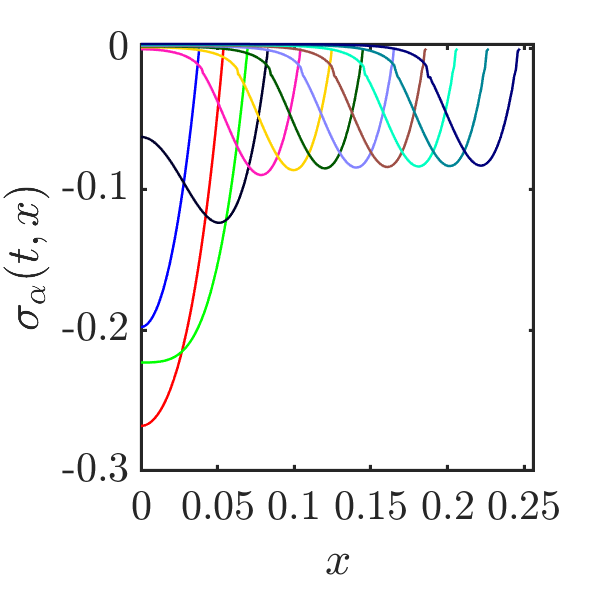}
				\label{fig:fs_vf2}
			\end{subfigure}
		} 
	\end{minipage}
	\begin{minipage}{0.99\textwidth}
		\centering
		\vspace{0.2cm}
		\includegraphics[scale=0.6]{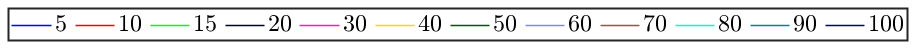}
	\end{minipage}
	\caption{Tumour growth in free suspension:\;
		The numerical results show how the nutrient concentration $(c(t,x))$, cell volume fraction $(\alpha(t,x))$, cell velocity $(u_{\alpha}(t,x))$, and cell stress $(\sigma_{\alpha}(t,x))$ evolve over time when the tumour is grown in free suspension. Simulation results are plotted at times $t \in \{5,10,15,20,30,40,50,60,70,80,90,100\}$ and were generated using the default parameter values stated in Table~\ref{tab:model_param}. }
	\label{fig:fres_sus_gr}
\end{figure}

At early times, the tumour is small and nutrient levels are everywhere sufficiently highly to enable cell proliferation. As a result, for $0 < t < 15$, the tumour cell velocity increases
monotonically with $x$, driving net tumour growth. At the same time, since the cells are viscous, they are resistant to movement. Consequently, the stress in the cell phase becomes
more compressive, especially near the tumour centre. 

As the tumour increases in size, nutrient levels at its centre, $x=0$, decrease until eventually they are so low that 
that cells start to die ($t >  20$). Cell death in the central necrotic core leads to a reduction in the tumour cell volume fraction there (dead cells are converted to fluid), relieving the compressive stress experienced by the tumour cells. 
By contrast, in the outer, nutrient-rich region, tumour cells continue to proliferate and drive tumour growth.
Since the tumour's net proliferation rate decreases when the necrotic core forms, the cell velocity on $x=\ell(t)$ and, hence, the tumour's overall growth rate decrease slightly, although both remain positive.
At long times, the outer proliferating rim attains a fixed width, at which the rate at which nutrient is supplied by diffusion balances the rate at which it is consumed by proliferating tumour cells. 
At the same time, the size of the central necrotic core increases,  and the cell velocity and stress there are negligible. 
Between the central necrotic core and the outer proliferating rim, an intermediate region forms in which the cells proliferate.
As a result of their resistance to movement, the cells in the intermediate region experience a compressive stress.    
Since the tumour is growing in free suspension, there is no mechanical resistance to its expansion
and, so, the magnitude of the compressive stress decreases from its minimum value, which is attained in the intermediate region,
towards the outer tumour boundary. 

As mentioned above, at long times, the cell velocity on the tumour boundary $x=\ell(t)$ evolves to a constant, positive value, indicating sustained and constant rate of tumour growth. This is consistent with the plot in Figure~\ref{fig:rad_time_theta} which 
shows that when $\vartheta = \nu = 0$ (i.e., when the tumour is grown in free suspension) the tumour boundary increases at a constant
rate. 

Taken together, the simulations results presented in Figure~\ref{fig:fres_sus_gr} and Figure~\ref{fig:rad_time_theta} show 
that the tumour's growth is limited by nutrient availability when it is cultured in free suspension. 

\subsection{Tumour growth limited by mechanical stress} 
\label{sec:small_tum}

\begin{figure}[htp]
	\begin{minipage}{0.94\textwidth}
		\resizebox{.95\textwidth}{!}{	
			\begin{minipage}{0.5cm}
				\centering
				\rotatebox{90}{\parbox{4cm}{\color{black} \centering  $\vartheta = 0.5$}}
			\end{minipage}
			\begin{minipage}{14cm}	
				\centering
				\begin{subfigure}[t]{0.33\textwidth}
					\caption{}
					\includegraphics[scale=0.5]{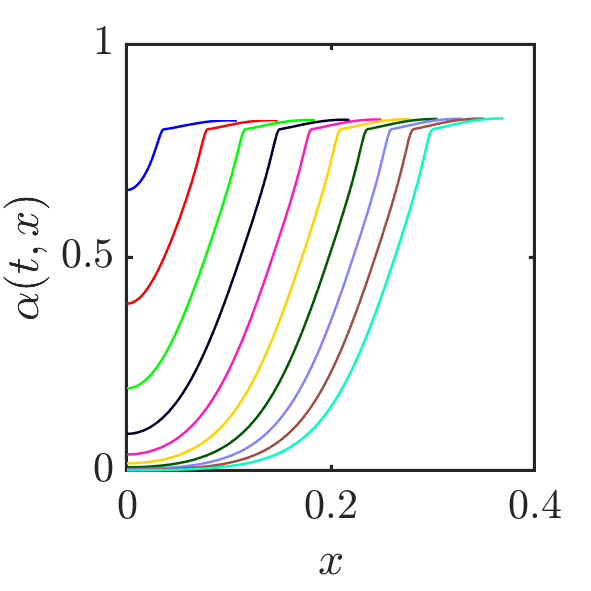}
					\label{fig:vf0p1}
				\end{subfigure}
				\begin{subfigure}[t]{0.33\textwidth}
					\caption{}
					\includegraphics[scale=0.5]{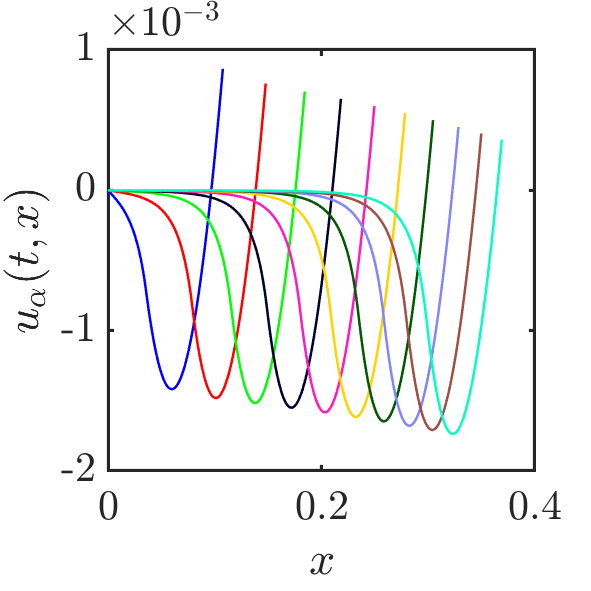}
					\label{fig:vel0p1}
				\end{subfigure}
				\begin{subfigure}[t]{0.32\textwidth}
					\caption{}
					\includegraphics[scale=0.5]{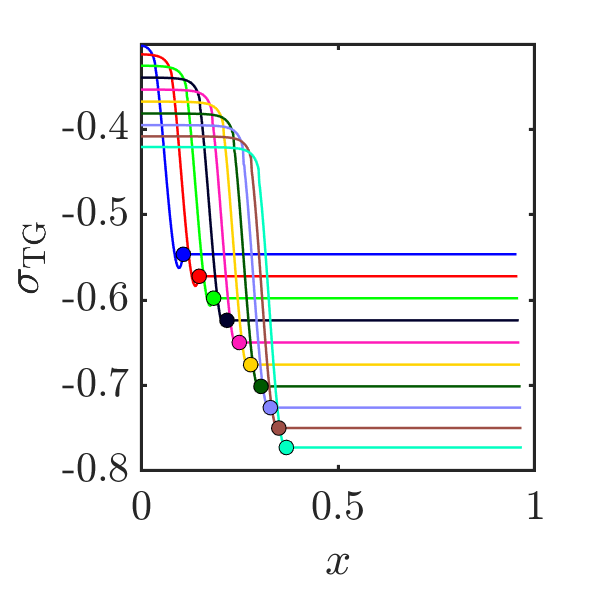}
					\label{fig:th_s1}
				\end{subfigure}
			\end{minipage}
		}
		\resizebox{.95\textwidth}{!}{	
			\begin{minipage}{0.2cm}
				\centering
				\rotatebox{90}{\parbox{4cm}{\color{black} \centering  $\vartheta = 2$}}
			\end{minipage} 
			\begin{minipage}{14cm}	
				\centering
				\begin{subfigure}[t]{0.32\textwidth}
					\caption{}
					\includegraphics[scale=0.5]{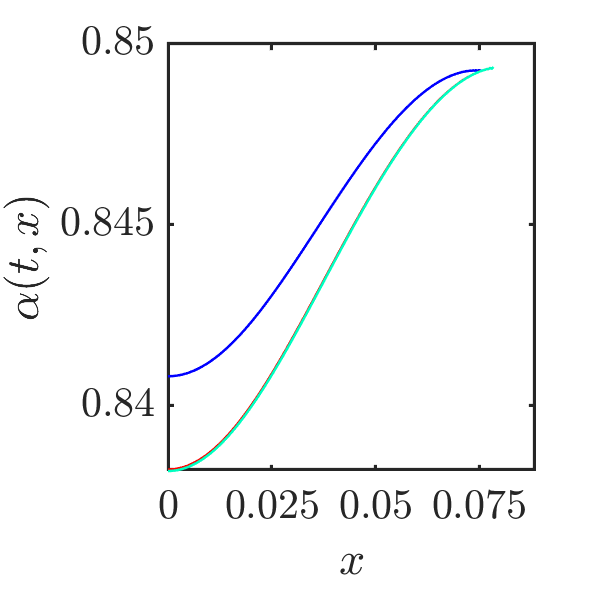}
					\label{fig:vf1}
				\end{subfigure}
				\begin{subfigure}[t]{0.32\textwidth}
					\caption{}
					\includegraphics[scale=0.5]{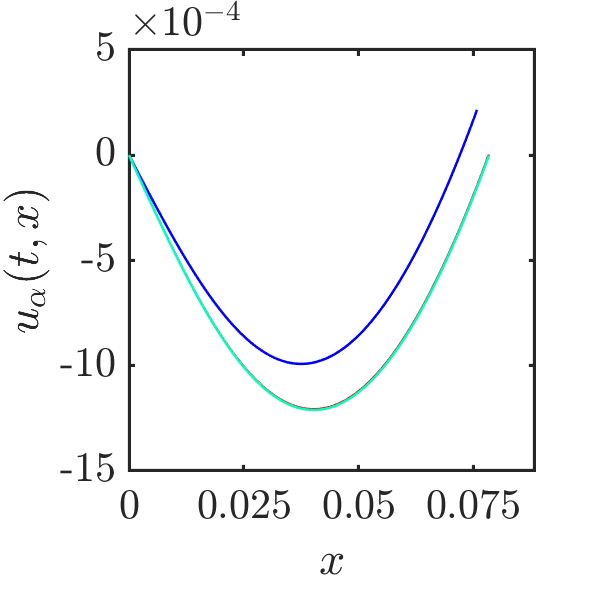}
					\label{fig:vel0p2}
				\end{subfigure}
				\begin{subfigure}[t]{0.32\textwidth}
					\caption{}
					\includegraphics[scale=0.5]{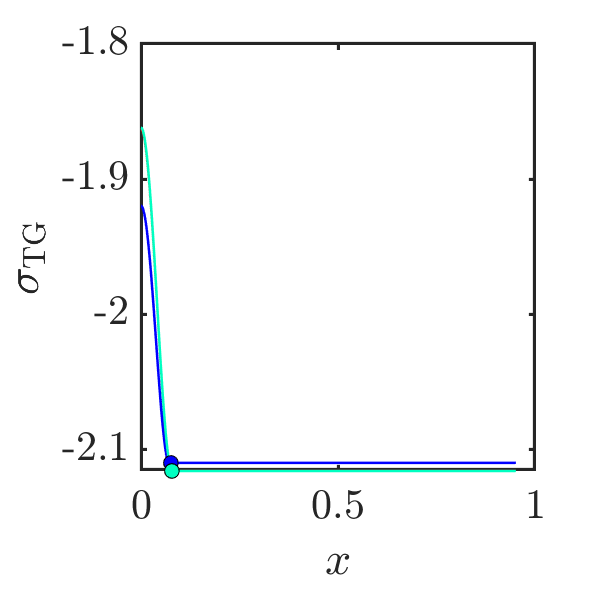}
					\label{fig:th_s2}
				\end{subfigure}
			\end{minipage}
		}
		\resizebox{.95\textwidth}{!}{	
			\begin{minipage}{0.4cm}
				\centering
				\rotatebox{90}{\parbox{3cm}{ \color{black} \centering $\vartheta = 20$}}
			\end{minipage}
			\begin{minipage}{14cm}	
				\centering
				\begin{subfigure}[t]{0.33\textwidth}
					\caption{}
					\includegraphics[scale=0.5]{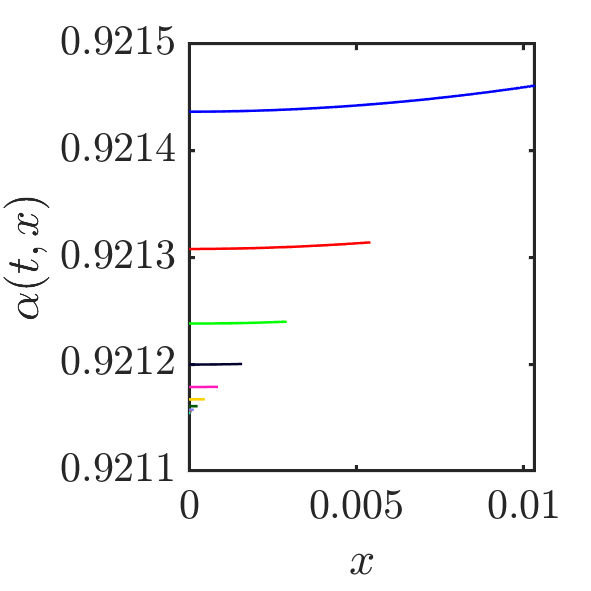}
					\label{fig:vf_20}
				\end{subfigure}
				\begin{subfigure}[t]{0.33\textwidth}
					\caption{}
					\includegraphics[scale=0.5]{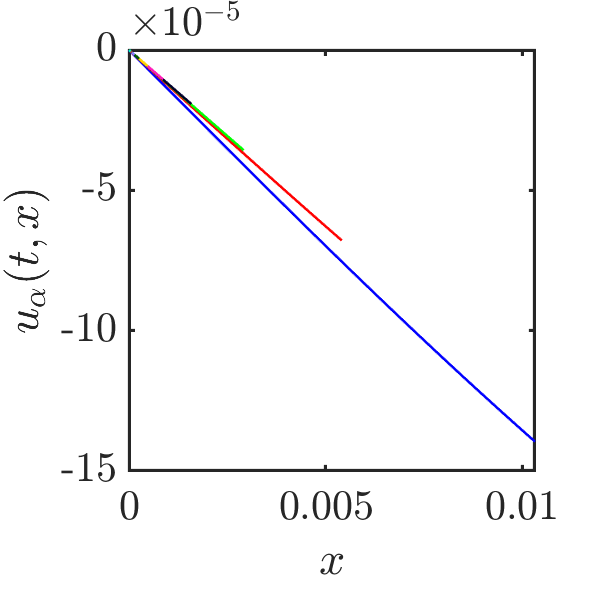}
					\label{fig:vel_20}
				\end{subfigure}
				\begin{subfigure}[t]{0.32\textwidth}
					\caption{}
					\includegraphics[scale=0.5]{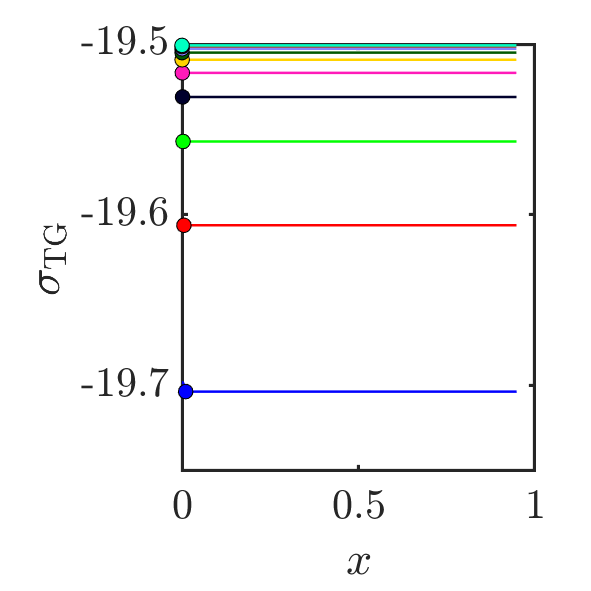}
					\label{fig:th_s3}
				\end{subfigure}
			\end{minipage}
		}
	\end{minipage}
	\begin{minipage}{3cm}
		\centering 
		Time values:
	\end{minipage}
	\begin{minipage}{14cm}
		\begin{subfigure}[c]{\textwidth}
			\centering	\includegraphics[scale=0.6]{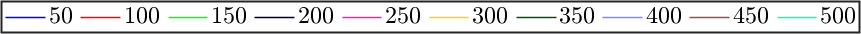}
		\end{subfigure}
	\end{minipage}
	
	\caption{Tumour growth limited by mechanical stress:\;
		Numerical results showing how the growth dynamics of a tumour embedded in a hydrogel change as the material properties of the hydrogel vary (all other model parameters are fixed at the default values stated in Table~\ref{tab:model_param}). 	
		For fixed values of $\vartheta$ ($\vartheta = 0.5,\,2.0$ and $20$), we show how the spatial distribution of the cell volume fraction $(\alpha(t,x))$ (Figures~\ref{fig:vf0p1},~\ref{fig:vf1},~\ref{fig:vf_20}), cell velocity $(u_{\alpha}(t,x))$ (Figures~\ref{fig:vel0p1},~\ref{fig:vel0p2},~\ref{fig:vel_20}), and mechanical stress $(\sigma_{\rm TG}(t,x))$ (Figures~\ref{fig:th_s1},~\ref{fig:th_s2},~\ref{fig:th_s3}) evolve over time. Note that the axes in each sub-plot is scaled differently.}
	\label{fig:effect_theta}
\end{figure}

\begin{figure}[h!]
	\centering
	\includegraphics[scale=1]{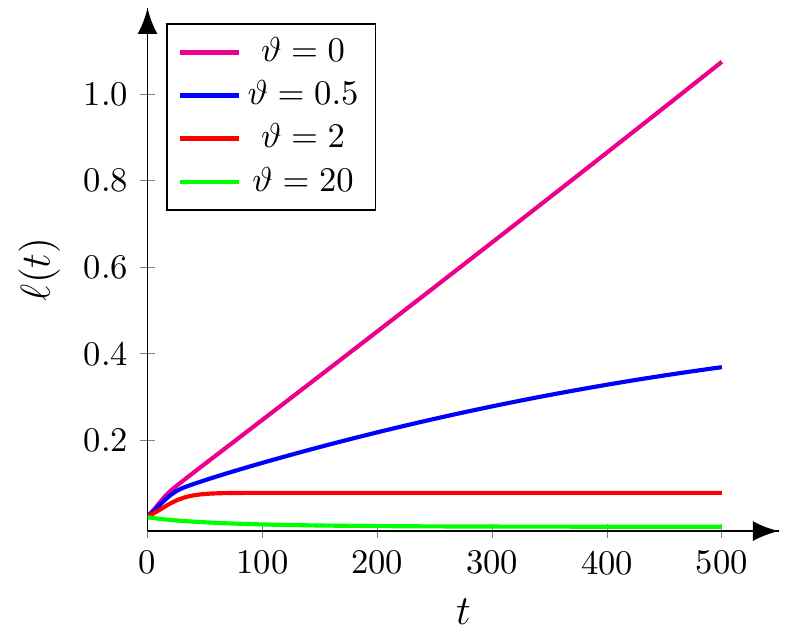}
	\caption{Tumour lengths over time for growth under mechanical stress:\; In all simulations, $\nu = 0$ and $\vartheta$ varies as indicated ($\vartheta = 0.0, 0.5, 2, 20$.)}
	\label{fig:rad_time_theta}
\end{figure}
We simulate tumour growth in a hydrogel by fixing the values of $\vartheta $ and $\nu$ so 
that the tumour experiences mechanical resistance to its growth ({\em i.e.}, $\sigma^{\mathrm{H}} < 0$).  
In what follows, we fix $\nu  = 0$ and vary $\vartheta$  between $\vartheta = 0.5$ and $\vartheta=20$. 
The numerical results presented in Figure~\ref{fig:effect_theta} 
show how the tumour cell volume fraction, $\alpha(t,x)$, the tumour cell velocity, $u_\alpha(t,x)$, and the mechanical stresses in the tumour and hydrogel,
$\sigma_{\rm TG} (t,x)$ change over time, where
\begin{align*}
\sigma_{\rm TG} = \begin{cases}
\sigma_\alpha\;\;\text{ if }\;\;0 \le x \le \ell(t), \\
\sigma^{\mathrm{H}}\;\text{ if }\;\;\ell(t) \le x \le 1.
\end{cases}
\end{align*}
The time evolution of the nutrient concentration is not plotted since it is similar to that shown in Figure~\ref{fig:nutr00}.  
Figure~\ref{fig:rad_time_theta} shows how the tumour length evolves over time.
We note that the tumour and hydrogel stresses are continuous at the tumour boundary, as indicated by the coloured circles on the line plots in Figures~\ref{fig:th_s1},~\ref{fig:th_s2}, and~\ref{fig:th_s3}. 
We note also that since $\sigma_{\alpha} < 0$ at $x = \ell(t)$, cells on the tumour boundary are mechanically compressed by the hydrogel. By contrast, when the tumour is grown in free suspension, $\sigma_{\alpha} = 0$ at $x = \ell(t)$  (see Figure~\ref{fig:fs_vf2}).  We remark that axes in each sub-plots in  Figure~\ref{fig:effect_theta} are scaled differently to resolve evolution of the corresponding variables more clearly. 
\medskip

The simulation results presented in Figure~\ref{fig:effect_theta} show that, when $\nu=0$, the tumour exhibits three qualitatively different behaviours as the hydrogel stiffness parameter $\vartheta$ varies:

\medskip
\noindent \textbf{Small $\vartheta$ ($0 \le \vartheta \le 1$):} the tumour grows to an equilibrium size that supports the formation of a necrotic core. 
The profiles for the tumour cell volume fraction, $\alpha(t,x)$, and velocity, $u_\alpha(t,x)$, are qualitatively similar to those for tumour growth in free
suspension (compare Figure~\ref{fig:effect_theta}(a,b) and Figure~\ref{fig:fres_sus_gr}).  However, the stress profiles are different. 
When the tumour is embedded in hydrogel, the entire tumour is under compressive stress ($\sigma_\alpha(t,x) < 0$ for $0 \leq x \leq \ell(t)$).
In the necrotic core, the stress is spatially uniform and increases in magnitude over time; this contrasts with the behaviour in free suspension, where the cell stress is zero in the necrotic core.
The mechanical stress is  most compressive on the tumour boundary. 
As the tumour increases in size and the hydrogel becomes more compressed, the stress that it exerts on the tumour boundary increases in magnitude, causing the tumour cell velocity there to decrease.  At long times, the cell velocity on $x = \ell(t)$ approaches zero and the tumour evolves 
to a steady state, with $0 < \ell_{\rm in} = \ell(t=0) <  \ell(\infty) < 1$ (see Figure~\ref{fig:rad_time_theta}). In this case, the tumour's growth is limited by the combined effects of
nutrient availability and the compressive forces it experiences from the hydrogel. 

\medskip
\noindent \textbf{Intermediate $\vartheta$  ($1 \le \vartheta \le 11$): } the tumour grows but, with increased compressive stress from the surrounding hydrogel, it rapidly evolves to a non-trivial equilibrium whose small size precludes the formation of a necrotic core because the nutrient concentration throughout the tumour exceeds the threshold for necrotic cell death (see the second row of Figure~\ref{fig:effect_theta}). 
In this case, the tumour's growth dynamics are dominated by inhibitory mechanical feedback from the hydrogel. 

\medskip
\noindent \textbf{Large $\vartheta$ ($\vartheta > 11$):} in this case, the stress in the hydrogel is extremely large in magnitude and compressive in nature. The compressive stress is so large that it prevents growth-induced tumour expansion and leads to  to its eventual elimination ($\ell (t) \rightarrow 0$ as $t \rightarrow \infty$) (see Figure~\ref{fig:rad_time_theta}).  
Here, the tumour's growth dynamics are dominated by inhibitory mechanical effects from the hydrogel. 

\medskip
We now investigate how the tumour's growth dynamics change as the hydrogel stiffness parameters vary. 
In Figures~\ref{fig:nu_const} and~\ref{fig:theta_const}, we show how the evolution of the tumour's outer boundary $x=\ell(t)$, changes when the 
parameters $\vartheta$ and $\nu$ are varied. As $\vartheta$ increases, the tumour's equilibrium size decreases but the timescale on which 
it attains its steady state remains constant, to leading order. By contrast, increasing $\nu$ has a dual effect: the tumour evolves more slowly to a larger steady state configuration.

\begin{figure}[h!]
	\begin{subfigure}[c]{0.45\textwidth}
		\caption{}
		\includegraphics[scale=0.9]{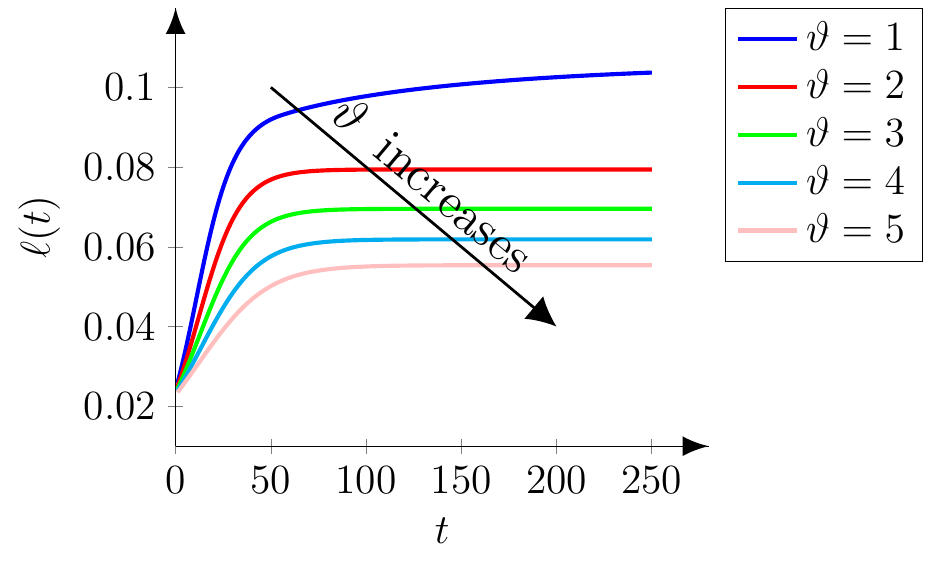}
		\label{fig:nu_const}
	\end{subfigure}
	\hspace{1cm}
	\begin{subfigure}[c]{0.45\textwidth}		
		\caption{}
		\includegraphics[scale=0.9]{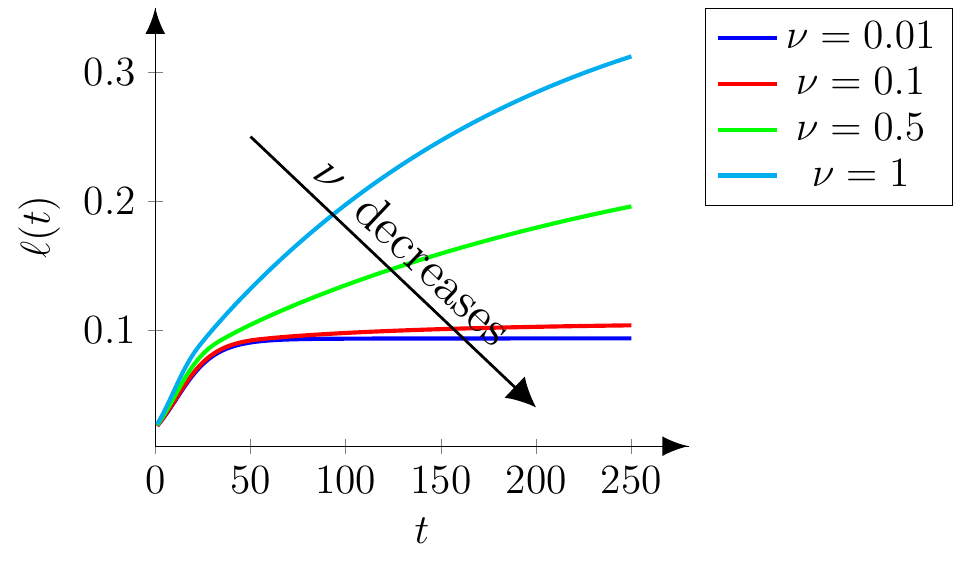}
		\label{fig:theta_const}
	\end{subfigure}
	\newline
	\centering
	\begin{subfigure}[c]{0.5\textwidth}
		\caption{}
		\centering
		\includegraphics[scale=0.6]{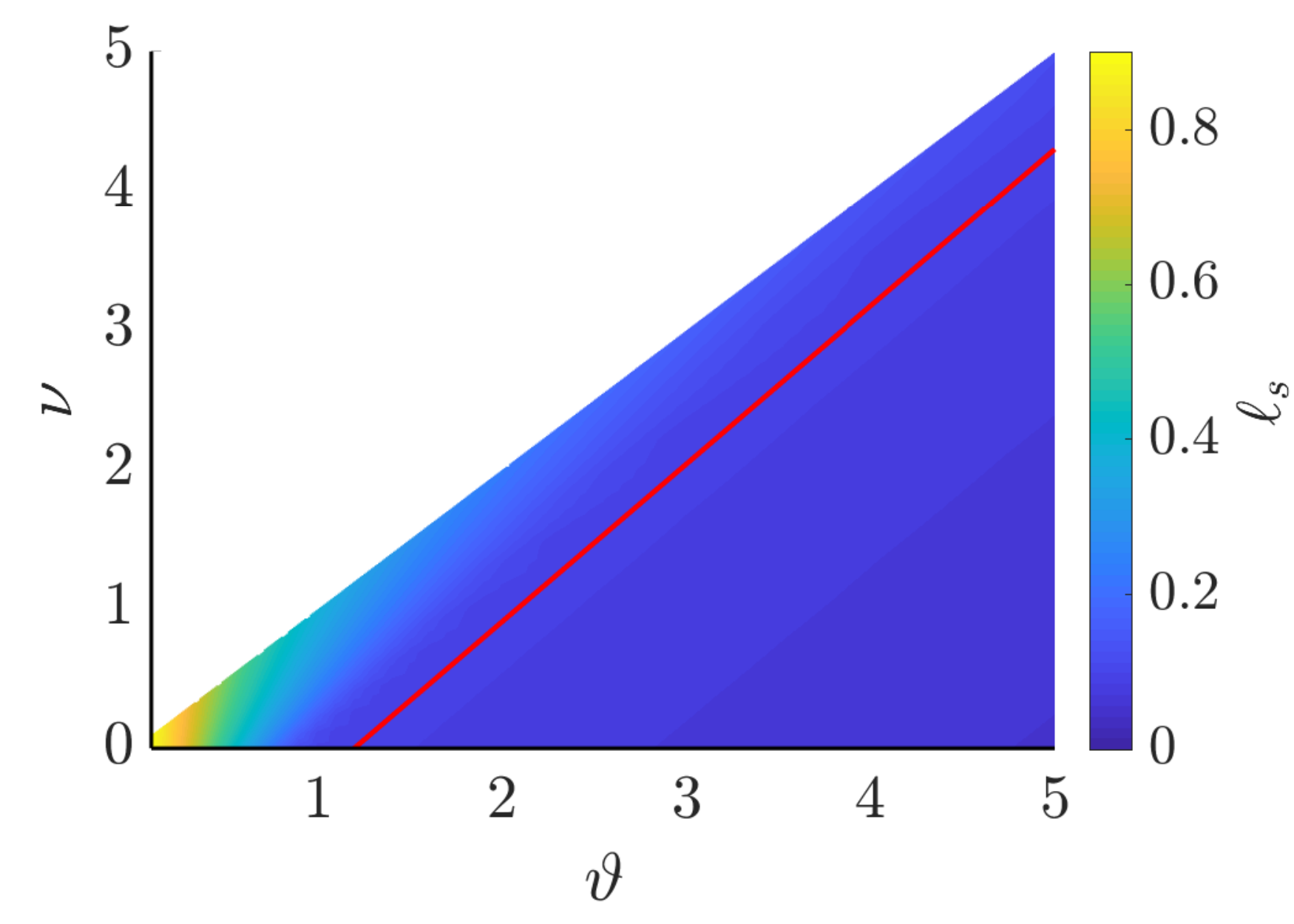}
		\label{fig:nutheta}
	\end{subfigure}
	\caption{Series of plots showing how the tumour's growth dynamics, notably its length $x = \ell(t)$, change as parameters governing the mechanical properties of the
		hydrogel in which it is embedded are varied. (a) When $\nu = 0.1$, the tumour's equilibrium size decreases as $\vartheta$ increases but the timescale
		over which it attains its equilibrium remains approximately constant. (b) When $\vartheta = 1$, the tumour's equilibrium size increases and the timescale on which it relaxes
		to its equilibrium also increases. (c) Diagram showing how the equilibrium tumour length changes as $\vartheta$ and $\nu$ vary (with $0 \leq \nu \leq \vartheta$). 
		For parameter values above the red line, the equilibrium tumour possesses a central necrotic core while for parameter values below the red line it does not.}
\end{figure}
For completeness, Figure~\ref{fig:nutheta} shows how the equilibrium tumour size, $x = \ell_s = \lim_{t \rightarrow \infty} \ell(t)$, changes when 
both $\nu$ and $\vartheta$ are varied, with $\vartheta \geq \nu \geq 0$. We estimate $\ell_s$ by solving the governing equations for $0 < t \leq T = 1000,$ and fixing 
$\ell_s = \ell(t=T)$.  
For $(\vartheta, \nu)$ combinations that lie above the red line, the tumour attains an equilibrium size with $\ell_s  \geq 0.01$. 
In such cases, insufficient nutrient reaches the central tumour region and a necrotic core forms. 
For $(\vartheta, \nu)$ combinations below the red line, the equilibrium tumour size is sufficiently small that nutrient levels throughout the tumour exceed the threshold for necrotic cell death.
We note also that the tumour only attains an equilibrium size with $\ell_s = 0.9$ for very small values of $\vartheta$ (towards the bottom--left corner in Figure~\ref{fig:nutheta}). Otherwise, the tumour attains an equilibrium size less than $\ell_s \approx 0.4$, indicating the influence of external compressive 
stress in controlling the final tumour size.  
In Figure~\ref{fig:nutheta}, parameter values for which tumour elimination occurs are not shown, in order to better resolve values of  
$(\vartheta,\nu)$ for which a necrotic core forms. However, extensive numerical simulations indicate that tumour elimination occurs for $\vartheta > 11$.

\section{Model analysis}  \label{sec:tumour_decay}

In this section we use asymptotic analysis to study the long time behaviour of our model and
to identify distinct parameter regimes in which the tumour is eliminated or attains a non-trivial steady state. 

The numerical results presented in Figures~\ref{fig:vf_20} and ~\ref{fig:vel_20} show that, when the hydrogel is very stiff, 
spatial variation in the cell volume is negligible, and the cell velocity is small and negative, and decreases linearly with distance from the tumour centre. Further,
the tumour length is small.
In order to investigate this behaviour, 
we define $\epsilon = \ell_{\rm in}$ and, guided by the parameter values stated in Table~\ref{tab:model_param}, we assume that 
\begin{align}
0 < \ell_{\rm in}  = \epsilon \ll 1 \text{ and } 0 < \epsilon^2 \dfrac{k}{\mu} \ll 1.
\end{align}
We seek approximate solutions of the governing equations which are regular power series expansions in the small parameter $\epsilon$:
\begin{align}
\alpha(t,x) ={}& \alpha_0(t) + \epsilon \alpha_1(t,\rho) + \mathcal{O}(\epsilon^2), \\
u_{\alpha}(t,x) ={}& \epsilon u_1(t,\rho) + \mathcal{O}(\epsilon^2), \\
c(t,x) ={}& 1 + \mathcal{O}(\epsilon^2) , \\
\ell(t) ={}& \epsilon \ell_1(t) + \mathcal{O}(\epsilon^2).
\end{align}
In the above expansions, we rescale the spatial coordinates so that $\rho = x /\epsilon$. 
We substitute the above expansions in Equations~\eqref{eqn:vf_dless}--\eqref{eqn:bvel_dless} and equate terms of $\mathcal{O}(\epsilon^n)$ ($n=0, 1$)
to arrive at the following equations for $\alpha_0$, $u_1$ and $\ell_1$:

\begin{align} \label{eqn:ldng-odr-cvf}
\dfrac{\mathrm{d} \alpha_0}{\mathrm{d} t} + \alpha_0 \dfrac{\partial u_1}{\partial \rho} ={}& \alpha_0(1 - \kappa - \alpha_0),\\ \label{eqn:ldng-odr-cv}
\dfrac{\partial^2 u_1}{\partial \rho^2} ={}& 0, \\
\dfrac{\mathrm{d} \ell_1}{\mathrm{d} t} ={}& u_1(t,\ell_1(t)). \label{eqn:ldng-odr-ell} \\
\end{align}
Here, we recall that $\kappa = (s_2 + s_3)/(1 + s_4)$, where $s_2,s_3$ and $s_4$ are non-negative parameters that control the death rate of the tumour cells (see Equation~(\ref{eqn:cv_dless})).
It is straightforward to show that, at leading order, Equations~\eqref{eqn:in_cond}--\eqref{eqn:bdr_cond_1} supply the following boundary and initial conditions for 
$\alpha_0$, $u_1$ and $\ell_1$:
\begin{gather}
\alpha_0(0) = \alpha_{\rm ini},\;\;  \\ \label{eqn:bc_1}
u_1(t,0) ={} 0, \\ \label{eqn:bc_2} 
\mu \dfrac{\partial u_1}{\partial \rho} (t, \ell_1(t)) = (\nu - \vartheta) + \dfrac{(\alpha_0(t) -\alpha^{\ast})^{+}}{(1 - \alpha_0(t))^2}, \\
\ell_1(0) = 1. 
\end{gather}
Integrating Equation~\eqref{eqn:ldng-odr-cv}, subject to boundary conditions~\eqref{eqn:bc_1} and~\eqref{eqn:bc_2}, yields the following
expression for $u_1(t,\rho)$:
\begin{align} \label{eqn:u1_steady}
u_1(t,\rho) = \dfrac{1}{\mu}\left((\nu - \vartheta) + \dfrac{(\alpha_0(t) -\alpha^{\ast})^{+}}{(1 - \alpha_0(t))^2}\right) \rho. 
\end{align}
Substituting for $u_1$ in Equations~\eqref{eqn:ldng-odr-cvf}  and \eqref{eqn:ldng-odr-ell}, we obtain the following 
differential equations for $\alpha_0$ and $\ell_1$:
\begin{align}\label{eqn:alpha0_steady}
\dfrac{\mathrm{d} \alpha_0}{\mathrm{d} t}  ={}& \alpha_0\left(\alpha_c - \alpha_0  - \dfrac{(\alpha_0 -\alpha^{\ast})^{+}}{\mu (1 - \alpha_0)^2} \right), \\
\label{eqn:ell1_steady}
\dfrac{\mu}{\ell_1} \dfrac{\mathrm{d} \ell_1}{\mathrm{d} t} ={}& (\nu - \vartheta) + \dfrac{(\alpha_0 -\alpha^{\ast})^{+}}{(1 - \alpha_0)^2},
\end{align}
where the critical volume fraction $\alpha_{\rm c}$ is defined as follows:
\begin{align}
\alpha_{\rm c} := 1 - \kappa + \dfrac{(\vartheta - \nu)}{\mu}.
\end{align}
Recall that we have assumed $\vartheta > \nu$ so that the hydrogel exerts a compressive force on the tumour.  
Since $\kappa \in (0,1)$, this ensures that $\alpha_{\rm c} > 0$. For a given value of $\alpha_{\rm c}$, we denote the steady state values of $\alpha_0(t)$ and $\ell_{1}(t)$ by $\alpha_{0\infty} := \lim_{t \rightarrow \infty} \alpha(t)$ and $\ell_{1\infty} := \lim_{t \rightarrow \infty} \ell(t)$. 

\medskip \noindent 
By setting $\frac{\mathrm{d}}{\mathrm{d}t} = 0$ in Equation~\eqref{eqn:alpha0_steady}, it is straightforward to show that there are two possible steady state solutions for $\alpha_0(t)$.

\paragraph{The trivial steady state.} 
If $\alpha_{0\infty} = 0$, then  Equation~\eqref{eqn:ell1_steady} yields
\begin{align}
\lim_{t \rightarrow \infty} \dfrac{\mu}{\ell_1} \dfrac{\mathrm{d} \ell_1}{\mathrm{d}t} = \nu -\vartheta < 0,
\end{align}
since $\nu < \vartheta$. We conclude that, in this case, $\ell_1(t)$ decays to zero at long times.  In practice, however, this case is \emph{physically unrealistic}. 
Since, nutrient levels inside a small tumour will be high (and approximately equal to the nutrient levels in the hydrogel), the tumour cells will proliferate rapidly and their volume fraction will be non-zero. 
Henceforth, we reject further consideration of this case. 

\paragraph{The non-trivial steady state.} 
The non-trivial steady state tumour cell volume fraction, $\alpha_{0 \infty}$  satisfies the following condition:
\begin{equation}
\label{eqn:alpha_c_stst}
\alpha_c - \alpha_{0\infty}  = \dfrac{(\alpha_{0\infty} -\alpha^{\ast})^{+}}{\mu (1 - \alpha_{0\infty})^2} \equiv \dfrac{\mathscr{H}(\alpha_{0\infty})}{\mu}\ge 0.
\end{equation}
\noindent There are two cases to consider, depending on whether $\alpha_{0 \infty} \le \alpha^\ast$ or $\alpha_{0 \infty} > \alpha^\ast$. 

\medskip \noindent 
\textbf{Case 1:} If $\alpha_{0 \infty} \le \alpha^\ast$, then $\mathscr{H}(\alpha_{0\infty}) =  0$ and Equation~\eqref{eqn:alpha0_steady} supplies $\alpha_{0\infty} = \alpha_{\rm c}$.  
Further, by referring to Equation~\eqref{eqn:ell1_steady}, we deduce that the tumour will be eliminated at rate $- (\vartheta - \nu)$.   In this case, $\alpha_{\rm c} \le \alpha_{\ast}$ is a necessary condition for $\alpha_{0\infty} = \alpha_{\rm c}$ (see, also, Figure~\ref{fig:interplot}).  
\begin{figure}[h!]
	\centering
	\begin{subfigure}{0.47\textwidth}
		\centering
		\caption{}
		\includegraphics[scale=0.9]{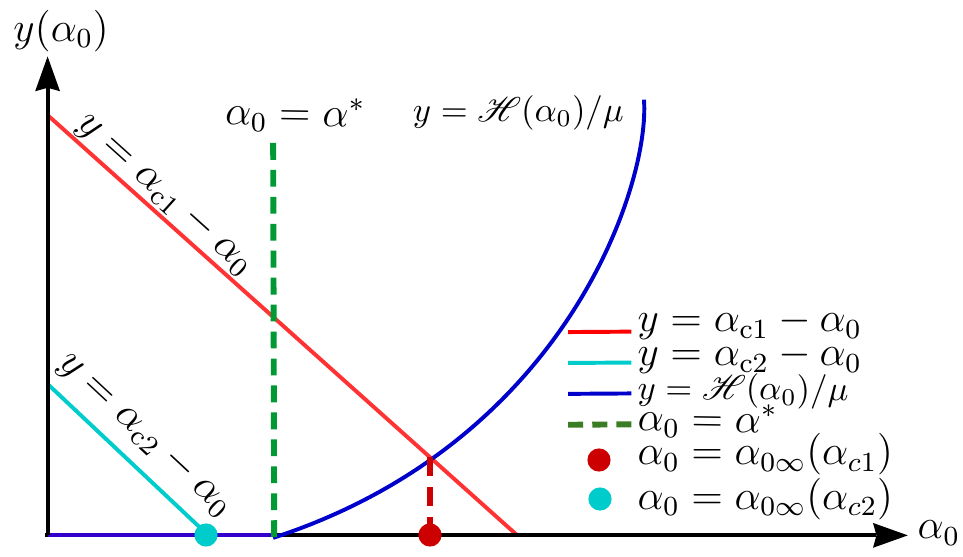}
		\label{fig:interplot}
	\end{subfigure}
	\hspace{0.7cm}
	\begin{subfigure}{0.47\textwidth}
		\caption{}
		\includegraphics[scale=1.2]{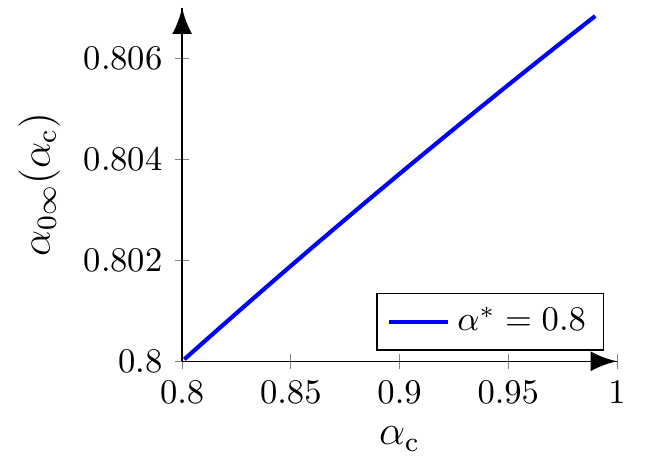}
		\label{fig:ac_steady}
	\end{subfigure}
	\label{fig:steady_alpha_0}
	\caption{ 
		Pair of plots showing how the non-trivial steady state value of the tumour cell volume fraction $\alpha_{0 \infty}$ changes as the parameter 
		$\alpha_{\rm c} = 1 - \kappa + (\vartheta - \nu)/\mu$ varies (see Equation~\eqref{eqn:alpha_c_stst}). 
		(a) The intersection of the curves $y = \mathscr{H}(\alpha_0)/\mu$ and $y = \alpha_{\mathrm{c}} - \alpha_0$ shows the steady state volume fraction $\alpha_{0\infty}(\alpha_{\rm c})$ for 
		$\alpha_{\rm c} \leq \alpha^\ast $ (cyan) and for $\alpha_{\rm c} > \alpha^\ast $ (red). (b) The steady state volume fraction $\alpha_{0 \infty}$ increases as $\alpha_{\rm c}$ increases.  
		In Figure~\ref{fig:ac_steady} we use Equation~(\ref{eqn:alpha_c_stst}) to show that $\alpha_{0\infty}$ increases monotonically with $\alpha_{\rm c} \in (\alpha^\ast,1)$ when $\alpha^\ast = 0.8$.}
\end{figure}

\medskip \noindent 
\textbf{Case 2:} If $\alpha_{0 \infty} > \alpha^{\ast}$, then the non-trivial steady state is the unique positive solution of
\begin{align}
\alpha_c - \alpha_{0\infty} = \dfrac{\alpha_{0\infty}  -\alpha^{\ast}}{\mu (1 - \alpha_{0\infty})^2}.
\end{align}
As shown in Figure~\ref{fig:interplot} for $\alpha_{\rm c} > \alpha^\ast$, 
the steady state volume fraction satisfies $ \alpha_{\ast} < \alpha_{0\infty} < \min(\alpha_{\rm c},1)$ and  Equation~\eqref{eqn:ell1_steady} supplies
\begin{align}
\lim_{t \rightarrow \infty}\dfrac{1}{\ell_1} \dfrac{\mathrm{d}\ell_1}{\mathrm{d}t} =  \alpha_{\rm c} - \alpha_{0\infty} -  \dfrac{(\vartheta - \nu)}{\mu} <  \alpha_c - \alpha^\ast - \dfrac{(\vartheta - \nu)}{\mu}.
\end{align}
Recalling that $\vartheta > \nu$, we deduce that if  $\alpha_{\rm c} < \alpha^\ast + (\vartheta -\nu)/\mu$, then the tumour decays. \\

In Figure~\ref{fig:ac_steady} we show that the non-trivial steady state $\alpha_{0\infty} = \alpha_{0\infty}(\alpha_{\rm c})$ 
is a monotonically increasing function of $\alpha_c$. To understand this behaviour, we proceed as follows. 
Choose $\alpha_{\rm c_1}$ and $\alpha_{\rm c_2}$ so that $\alpha_{\rm c_1} \ge \alpha_{\rm c_2}$ and suppose that $\alpha_{0\infty}(\alpha_{\rm c_1}) <  \alpha_{0\infty}(\alpha_{\rm c_2})$.
Rearranging Equation~\eqref{eqn:alpha_c_stst} we have that
\begin{align}
\alpha_{\rm c_1}  = 	\alpha_{0\infty}(\alpha_{\rm c_1}) + \dfrac{(\alpha_{0\infty}(\alpha_{\rm c_1}) -\alpha^{\ast})}{\mu (1 - \alpha_{0\infty}(\alpha_{\rm c_1}))^2}  < 
\alpha_{0\infty}(\alpha_{\rm c_2})  + \dfrac{(\alpha_{0\infty}(\alpha_{\rm c_2}) -\alpha^{\ast})}{\mu (1 - \alpha_{0\infty}(\alpha_{\rm c_2}))^2} = \alpha_{\rm c_2},
\end{align}
which contradicts our assumption that  $\alpha_{\rm c_1}\ge \alpha_{\rm c_2}$ and, so, we conclude that $\alpha_{0 \infty}$ increases monotonically with $\alpha_{\rm c}$.

With $\alpha_{\rm c} = 1 - \kappa + (\vartheta - \nu)/\mu$, it is now clear that both $\alpha_{\rm c}$ and $\alpha_{0 \infty}$ are increasing functions of the hydrogel stiffness parameter $\vartheta$.  
Recall that the tumour is a two-phase mixture of viscous cells and inviscid fluid. Since the fluid is less resistant to flow than the cells, as $\alpha_{\rm c}$ (or equivalently $\vartheta$) increases, 
the compressive force acting on the tumour increases, fluid is squeezed out of the tumour and the tumour cell volume fraction increases (see Figure \ref{fig:effect_theta}).  

We end this section with a remark on the tumour's long-time growth rate. 
If $u_\alpha(t,\ell(t)) \leq 0$ for all $t \geq 0$, then the tumour is eventually eliminated. By contrast, if $u_\alpha(t,\ell(t)) \geq 0$ for all $t \geq 0$, it is possible to prove 
that $u_\alpha(t,\ell(t)) \rightarrow 0$ when $\ell(t) \rightarrow 1$ (see Appendix~\ref{sec:boxdomain}). 
As the tumour expands, the polymer and water molecules in the hydrogel become increasingly packed until, eventually, the short-range repulsive forces between them cause the hydrogel to become effectively incompressible. Since the hydrogel and tumour occupy a finite, one-dimensional domain, incompressibility eventually leads to length invariance of the hydrogel and, hence, the tumour.
In practice, and as the numerical results in Figure 7 indicate, the tumour becomes mechanically inhibited from growth when $\ell (t) \approx 0.9$ (\emph{i.e.,} before the tumour reaches the domain 
boundary, at $x=1$).

\section{Discussion}
\label{sec:discussion}
We have presented a new mathematical model that describes the one-dimensional growth of an avascular tumour spheroid embedded within a deformable hydrogel. 
The tumour is viewed as a two-phase mixture of cells and fluid and the hydrogel is modelled as a compressible, hyperelastic material. 
Following the approach outlined in~\cite{IMA::breward_2002}, the principles of mass and momentum conservation 
are used to derive partial differential equations that describe the time evolution of
the volume fractions and velocities of the cell and fluid phases.
Following~\cite{MMNP::Lewin:2020} and generalising the approaches in~\cite{IMA::breward_2002} and~\cite{IMA::Chen2001}, 
the tumour and cell phases are viewed as viscous and inviscid fluids, respectively. 
Further, local levels of an externally-supplied, diffusible nutrient regulate local rates of tumour cell proliferation and death.
As the tumour increases in size, it deforms the hydrogel which, in turn, exerts a compressive stress on the tumour.
The balance between the expansive pressure caused by tumour growth and the resistance caused by hydrogel deformation determines
the tumour's growth dynamics. 

We presented numerical results which show how the tumour's growth dynamics change as the material parameters of the hydrogel are varied. For soft hydrogels,
the tumour grows to a large size, and its dynamics are similar to those associated with growth in free suspension, where nutrient availability 
is growth-rate limiting. For stiffer hydrogels, 
the tumour evolves to a non-trivial equilibrium whose size is limited by both nutrient availability and mechanical effects. For very stiff hydrogels, mechanical effects are dominant and the compressive forces exerted by the hydrogel on the tumour are so large that it is eliminated.  
The asymptotic analysis presented in Section~\ref{sec:tumour_decay} confirms the numerical results 
and enables us to identify parameter regimes in which tumour elimination is predicted.

There are many ways in which the model presented in this paper could be extended. 
In the \emph{in vitro} context, possible directions for future research include: 
extension of the governing equations to describe tumour growth in two and three spatial dimensions,
and treatment of the hydrogel as a two-phase, poroelastic material to enable explicit consideration of fluid flow between the tumour and hydrogel.
We could also account for behavioural changes in the tumour cells caused by mechanotransduction.
For example, tumour cells may respond to compressive mechanical stress by reducing their rate of cell proliferation
and increasing their rate of cell death~\cite{delarue,Roose:2003}. They may also release 
protein-digesting enzymes (proteases) that remodel the extracellular matrix in which they are embedded, reducing the 
mechanical stress that they experience, and creating space into which they can expand~\cite{park,clerk,franks2003a}.
Such model extensions would enable investigation of the interplay between tumour growth, hydrogel compression, and tumour-induced remodelling of the hydrogel. 

It would also be interesting to specialise the model to study tumour growth \emph{in vivo}: this could be achieved
by including additional phases to account for the vasculature~\cite{IMA::hubbard},
different types of immune cells (e.g., macrophages and T cells)~\cite{web2007,owen2004}
and stromal cells (e.g., fibroblasts), 
their cross talk with each other and the tumour cells, and also remodelling of the vasculature and extracellular matrix in which the tumour is located~\cite{dyson2016}.
By simulating these effects, we aim to increase understanding of the ways in which changes in the extracelluar matrix
impact tumour growth and responses to existing treatments such as radiotherapy, chemotherapy and immunotherapy~\cite{esfahni2020,thariat2013,chabner2005} 
while also identifying those that could be used to inhibit tumour progression~\cite{guissani}.   

Other directions that merit further consideration relate to validation and parameterisation of the model against experimental data, and also its comparison with
the earlier models proposed by~\cite{IMA::Chen2001, Roose:2003, Yan:2021}.
At present, the available experimental data are limited to measurements of tumour volume~\cite{Helmlinger1997}, spatial staining for cell proliferation and cell death and measurements of stress distribution in the hydrogel~\cite{rakesh1}.  In the numerical simulations presented in this paper, we have used parameter values from previous literature~\cite{IMA::breward_2002,IMA::breward_2001} and several interesting qualitative features of the model variables are illustrated. However, rigorous parameter estimation for this model is not yet undertaken and would be an exciting avenue for further work including data integration methods for incorporating multiple sources of data and model selection, both of which have become common in other data-driven areas of mathematical biology \cite{alahmadi2020influencing}.

\medskip
\noindent \textbf{Acknowledgements.} GCR acknowledges Monash University for funding. He also thanks N Nataraj and J Droniou for support during his PhD studies. 
We also thank Avner Friedman for his many, seminal contributions to mathematical biology and for being an inspirational role model to us and so many other researchers in the field. 


\bibliographystyle{unsrtnat}
\bibliography{thesis}

\begin{thebibliography}{61}
\providecommand{\natexlab}[1]{#1}
\providecommand{\url}[1]{\texttt{#1}}
\expandafter\ifx\csname urlstyle\endcsname\relax
  \providecommand{\doi}[1]{doi: #1}\else
  \providecommand{\doi}{doi: \begingroup \urlstyle{rm}\Url}\fi

\bibitem[Hirschhaeuser et~al.(2010)Hirschhaeuser, Menne, Dittfeld, West,
  Mueller-Klieser, and Kunz-Schughart]{hirschaeuser}
F.~Hirschhaeuser, H.~Menne, C.~Dittfeld, J.~West, W.~Mueller-Klieser, and L.~A.
  Kunz-Schughart.
\newblock Multicellular tumor spheroids: an underestimated tool is catching up
  again.
\newblock \emph{J. Biotec.}, 148\penalty0 (1):\penalty0 3--15, 2010.
\newblock URL \url{https://doi.org/10.1016/j.jbiotec.2010.01.012}.

\bibitem[Ferrara(2002)]{ferrara2002vegf}
N.~Ferrara.
\newblock \uppercase{VEGF} and the quest for tumour angiogenesis factors.
\newblock \emph{Nat. Rev. Cancer}, 2\penalty0 (10):\penalty0 795--803, 2002.
\newblock URL \url{https://doi.org/10.1038/nrc909}.

\bibitem[Duffy et~al.(2008)Duffy, McGowan, and Gallagher]{duffy2008}
M.~J. Duffy, P.~M. McGowan, and W.~M. Gallagher.
\newblock Cancer invasion and metastasis: changing views.
\newblock \emph{J. Pathol.}, 214\penalty0 (3):\penalty0 283--293, 2008.
\newblock URL \url{https://doi.org/doi:10.1002/path.2282}.

\bibitem[Bissell and Radisky(2001)]{bissell2001putting}
M.~J. Bissell and D.~Radisky.
\newblock Putting tumours in context.
\newblock \emph{Nat. Rev. Cancer}, 1\penalty0 (1):\penalty0 46--54, 2001.
\newblock URL \url{https://doi.org/10.1038/35094059}.

\bibitem[Gonzalez et~al.(2018)Gonzalez, Hagerling, and Werb]{gonzalez2018}
H.~Gonzalez, C.~Hagerling, and Z.~Werb.
\newblock Roles of the immune system in cancer: from tumor initiation to
  metastatic progression.
\newblock \emph{Genes Dev.}, 32\penalty0 (19-20):\penalty0 1267--1284, 2018.
\newblock URL \url{https://doi.org/10.1101/gad.314617.118}.

\bibitem[Chaudhuri et~al.(2018)Chaudhuri, Low, and
  Lim]{chaudhuri2018mechanobiology}
P.~K. Chaudhuri, B.~C. Low, and C.~T. Lim.
\newblock Mechanobiology of tumor growth.
\newblock \emph{Chem. Rev.}, 118\penalty0 (14):\penalty0 6499--6515, 2018.
\newblock URL \url{https://doi.otg/10.1021/acs.chemrev.8b00042}.

\bibitem[Jain et~al.(2014)Jain, Martin, and Stylianopoulos]{jain2014role}
R.~K. Jain, J.~D. Martin, and T.~Stylianopoulos.
\newblock The role of mechanical forces in tumor growth and therapy.
\newblock \emph{Annu. Rev. Biomed. Eng.}, 16:\penalty0 321, 2014.
\newblock URL \url{https://doi.org/10.1146/annurev-bioeng-071813-105259}.

\bibitem[Helmlinger et~al.(1997)Helmlinger, Netti, Lichtenbeld, Melder, and
  Jain]{Helmlinger1997}
G.~Helmlinger, P.~A. Netti, H.~C. Lichtenbeld, R.~J. Melder, and R.~K. Jain.
\newblock Solid stress inhibits the growth of multicellular tumor spheroids.
\newblock \emph{Nat. Biotec.}, 15\penalty0 (8):\penalty0 778--783, 1997.
\newblock URL \url{https://doi.org/10.1038/nbt0897-778}.

\bibitem[Cheng et~al.(2009)Cheng, Tse, Jain, and Munn]{rakesh1}
G.~Cheng, J.~Tse, R.~K. Jain, and L.~L. Munn.
\newblock Micro-environmental mechanical stress controls tumor spheroid size
  and morphology by suppressing proliferation and inducing apoptosis in cancer
  cells.
\newblock \emph{PLoS ONE}, 4\penalty0 (2):\penalty0 17, 2009.
\newblock URL \url{https://doi.org/10.1371/journal.pone.0004632}.

\bibitem[Delarue et~al.(2014)Delarue, Montel, Vignjevic, Prost, Joanny, and
  Cappello]{delarue}
M.~Delarue, F.~Montel, D.~Vignjevic, J.~Prost, J.~F. Joanny, and G.~Cappello.
\newblock Compressive stress inhibits proliferation in tumor spheroids through
  a volume limitation.
\newblock \emph{Biophys. J.}, 107\penalty0 (8):\penalty0 1821--1828, 2014.
\newblock URL \url{https://doi.org/10.1016/j.bpj.2014.08.031}.

\bibitem[Takao et~al.(2019)Takao, Taya, and Chiew]{Takaobio}
S.~Takao, M.~Taya, and C.~Chiew.
\newblock Mechanical stress--induced cell death in breast cancer cells.
\newblock \emph{Biol. Open.}, 8\penalty0 (8), 2019.
\newblock URL \url{https://doi.org/10.1242/bio.043133}.

\bibitem[Northcott et~al.(2018)Northcott, Dean, Mouw, and Weaver]{oncogenes}
J.~M. Northcott, I.~S. Dean, J.~K. Mouw, and V.~M. Weaver.
\newblock Feeling stress: the mechanics of cancer progression and aggression.
\newblock \emph{Front. Cell Dev. Biol.}, 6:\penalty0 17, 2018.
\newblock URL \url{https://doi.org/10.3389/fcell.2018.00017}.

\bibitem[Liu et~al.(2020)Liu, Luo, Ju, and Song]{liu2020}
Q.~Liu, Q.~Luo, Y.~Ju, and G.~Song.
\newblock Role of the mechanical microenvironment in cancer development and
  progression.
\newblock \emph{Cancer Biol. Med.}, 17\penalty0 (2):\penalty0 282--292, 2020.
\newblock URL \url{https://doi.org/10.20892/j.issn.2095-3941.2019.0437}.

\bibitem[Bull and Byrne(2022)]{bull2022hallmarks}
J.~A. Bull and H.~M. Byrne.
\newblock The hallmarks of mathematical oncology.
\newblock \emph{Proceedings of the IEEE}, 110\penalty0 (5):\penalty0 523--540,
  2022.
\newblock URL \url{https:/doi.org/10.1109/JPROC.2021.3136715}.

\bibitem[Flegg and Nataraj(2019)]{flegg2019mathematical}
J.~A. Flegg and N.~Nataraj.
\newblock Mathematical modelling and avascular tumour growth.
\newblock \emph{Resonance}, 24\penalty0 (3):\penalty0 313--325, 2019.
\newblock URL \url{https://doi.org/10.1007/s12045-019-0782-8}.

\bibitem[Mathur et~al.(2022)Mathur, Barnett, Scher, and
  Xavier]{mathur2022optimizing}
D.~Mathur, E.~Barnett, H.~I. Scher, and J.~B Xavier.
\newblock Optimizing the future: how mathematical models inform treatment
  schedules for cancer.
\newblock \emph{Trends in Cancer}, 8\penalty0 (6):\penalty0 506--516, 2022.
\newblock URL \url{https://doi.org/10.1016/j.trecan.2022.02.005}.

\bibitem[Byrne et~al.(2006)Byrne, Alarcon, Owen, Webb, and
  Maini]{IMA::Byrne20061563}
H.~M. Byrne, T.~Alarcon, M.~R. Owen, S.~D. Webb, and P.~K. Maini.
\newblock Modelling aspects of cancer dynamics: A review.
\newblock \emph{Philo. Trans. Roy. Soc. A}, 364\penalty0 (1843):\penalty0
  1563--1578, 2006.
\newblock URL \url{10.1098/rsta.2006.1786}.

\bibitem[Roose et~al.(2007)Roose, Chapman, and Maini]{IMA::Roose2007179}
T.~Roose, S.~J. Chapman, and P.~K. Maini.
\newblock Mathematical models of avascular tumour growth.
\newblock \emph{SIAM Rev.}, 49:\penalty0 179--208, 2007.
\newblock URL \url{https://doi.org/10.1137/S0036144504446291}.

\bibitem[Bull et~al.(2020)Bull, Mech, Quaiser, Waters, and
  Byrne]{bull2020mathematical}
J.~A. Bull, F.~Mech, T.~Quaiser, S.~L. Waters, and H.~M. Byrne.
\newblock Mathematical modelling reveals cellular dynamics within tumour
  spheroids.
\newblock \emph{PLoS Comput. Bio.}, 16\penalty0 (8), 2020.
\newblock URL \url{https://doi.org/10.1371/journal.pcbi.1007961}.

\bibitem[Wallace and Guo(2013)]{wallace2013properties}
D.~Wallace and X.~Guo.
\newblock Properties of tumor spheroid growth exhibited by simple mathematical
  models.
\newblock \emph{Front. Oncol.}, 3:\penalty0 51, 2013.
\newblock URL \url{10.3389/fonc.2013.00051}.

\bibitem[Gatenby and Gawlinski(1996)]{gatenby1996}
R.~A. Gatenby and E.~T. Gawlinski.
\newblock {A reaction-diffusion model of cancer invasion}.
\newblock \emph{Cancer Research}, 56\penalty0 (24):\penalty0 5745--5753, 1996.
\newblock URL
  \url{https://aacrjournals.org/cancerres/article-pdf/56/24/5745/2462558/cr0560245745.pdf}.

\bibitem[Greenspan(1972)]{greenspan1972}
H.~P. Greenspan.
\newblock Models for the growth of a solid tumor by diffusion.
\newblock \emph{Stud. Appl. Math.}, 51\penalty0 (4):\penalty0 317--340, 1972.
\newblock URL \url{https://doi.org/10.1002/sapm1972514317}.

\bibitem[Chen and Friedman(2003)]{Friedman:SIAM:2003}
X.~Chen and A.~Friedman.
\newblock A free boundary problem for an elliptic-hyperbolic system: an
  application to tumor growth.
\newblock \emph{SIAM J. Math. Anal.}, 35\penalty0 (4):\penalty0 974--986, 2003.
\newblock URL \url{https://doi.org/10.1137/S0036141002418388}.

\bibitem[Fontelos and Friedman(2003)]{Friedman:2003}
M.~A. Fontelos and A.~Friedman.
\newblock Symmetry-breaking bifurcations of free boundary problems in three
  dimensions.
\newblock \emph{Asymptot. Anal.}, 35\penalty0 (3-4):\penalty0 187--206, 2003.
\newblock URL \url{http://www.jstor.org/stable/24902303}.

\bibitem[Chen et~al.(2005)Chen, Cui, and Friedman]{Friedman:2005}
X.~Chen, S.~Cui, and A.~Friedman.
\newblock A hyperbolic free boundary problem modeling tumor growth: Asymptotic
  behavior.
\newblock \emph{Trans. Amer. Maths. Soc.}, 357\penalty0 (12):\penalty0
  4771--4804, 2005.
\newblock URL \url{https:doi.org/10.1090/S0002-9947-05-03784-0}.

\bibitem[Chen and Friedman(2013)]{Friedman:2013}
D.~Chen and A.~Friedman.
\newblock A two-phase free boundary problem with discontinuous velocity:
  Application to tumor model.
\newblock \emph{J. Math. Anal. Appl.}, 399\penalty0 (1):\penalty0 378--393,
  2013.
\newblock URL \url{https://doi.org/10.1016/j.jmaa.2012.10.035}.

\bibitem[Byrne et~al.(2003)Byrne, King, McElwain, and Preziosi]{Byrne2003a}
H.~M. Byrne, J.~R. King, D.~L.~S. McElwain, and L.~Preziosi.
\newblock A two-phase model of solid tumour growth.
\newblock \emph{Appl. Math. Lett.}, 16\penalty0 (4):\penalty0 567--573, 2003.
\newblock URL \url{https://doi.org/10.1016/S0893-9659(03)00038-7}.

\bibitem[Ambrosi and Mollica(2004)]{Ambrosi2004}
D.~Ambrosi and F~Mollica.
\newblock The role of stress in the growth of a multicell spheroid.
\newblock \emph{J. Math. Biol.}, \penalty0 (48):\penalty0 477–499, 2004.
\newblock URL \url{https://doi.org/10.1007/s00285-003-0238-2}.

\bibitem[Ambrosi et~al.(2017)Ambrosi, Pezzuto, Riccobelli, Stylianopoulos, and
  Ciarletta]{Ambrosi2017}
D.~Ambrosi, S.~Pezzuto, D.~Riccobelli, T.~Stylianopoulos, and P.~Ciarletta.
\newblock Solid tumors are poroelastic solids with a chemo-mechanical feedback
  on growth.
\newblock \emph{J. Elast.}, \penalty0 (129):\penalty0 107–124, 2017.
\newblock URL \url{https://doi.org/10.1007/s10659-016-9619-9}.

\bibitem[Chen et~al.(2001)Chen, Byrne, and King]{IMA::Chen2001}
C.~Y. Chen, H.~M. Byrne, and J.~R. King.
\newblock The influence of growth-induced stress from the surrounding medium on
  the development of multicell spheroids.
\newblock \emph{J. Math. Biol.}, 43\penalty0 (3):\penalty0 191--220, 2001.
\newblock URL \url{https://doi.org/10.1007/s002850100091}.

\bibitem[Roose et~al.(2003)Roose, Netti, Munn, Boucher, and Jain]{Roose:2003}
T.~Roose, P.~A. Netti, L.~L. Munn, Y.~Boucher, and R.~K. Jain.
\newblock Solid stress generated by spheroid growth estimated using a linear
  poroelasticity model.
\newblock \emph{Microvascular Research}, 66\penalty0 (3):\penalty0 204--212,
  2003.
\newblock URL \url{https://doi.org/10.1016/S0026-2862(03)00057-8}.

\bibitem[Yan et~al.(2021)Yan, Ramirez-Guerrero, Lowengrub, and Wu]{Yan:2021}
H.~Yan, D.~Ramirez-Guerrero, J.~Lowengrub, and .M~Wu.
\newblock Stress generation, relaxation and size control in confined tumor
  growth.
\newblock \emph{PLOS Comput. Bio.}, 17\penalty0 (12), 2021.
\newblock URL \url{https://doi.org/10.1371/journal.pcbi.1009701}.

\bibitem[Breward et~al.(2002)Breward, Byrne, and Lewis]{IMA::breward_2002}
C.~J.~W. Breward, H.~M. Byrne, and C.~E. Lewis.
\newblock The role of cell-cell interactions in a two-phase model for avascular
  tumour growth.
\newblock \emph{J. Math. Biol.}, 45\penalty0 (2):\penalty0 125--152, 2002.
\newblock URL \url{https://doi.org/10.1007/s002850200149}.

\bibitem[Lemon et~al.(2006)Lemon, King, Byrne, Jensen, and
  Shakesheff]{Lemonetal}
G.~Lemon, J.~R. King, H.~M. Byrne, O.~E. Jensen, and K.~M. Shakesheff.
\newblock Mathematical modelling of engineered tissue growth using a multiphase
  porous flow mixture theory.
\newblock \emph{J. Math. Bio.}, 52:\penalty0 571--594, 2006.
\newblock URL \url{https://doi.org/10.1007/s00285-005-0363-1}.

\bibitem[Bergström(2015)]{BERGSTROM2015209}
J.~S. Bergström.
\newblock \emph{Mechanics of Solid Polymers}.
\newblock William Andrew, 2015.
\newblock URL \url{https://doi.org/10.1016/B978-0-323-31150-2.00005-4}.

\bibitem[Yeoh(1989)]{YEOH1989425}
O.~H. Yeoh.
\newblock Comprehensive polymer science and supplements.
\newblock Pergamon, Amsterdam, 1989.
\newblock URL \url{https://doi.org/10.1016/B978-0-08-096701-1.00251-2}.

\bibitem[Flory(1953)]{flory1953}
P.~J. Flory.
\newblock \emph{Principles of Polymer Chemistry}.
\newblock Baker lectures 1948. Cornell University Press, 1953.

\bibitem[Hong et~al.(2008)Hong, Zhao, Zhou, and Suo]{hong2008}
W.~Hong, X.~Zhao, J.~Zhou, and Z.~Suo.
\newblock A theory of coupled diffusion and large deformation in polymeric
  gels.
\newblock \emph{J. Mech. Phys. Solids}, 56\penalty0 (5):\penalty0 1779--1793,
  2008.
\newblock URL \url{10.1016/j.jmps.2007.11.010}.

\bibitem[Ward and R.~King(1997)]{ward_1}
J.~Ward and J.~R.~King.
\newblock Mathematical modelling of avascular-tumour growth.
\newblock \emph{IMA J. Math. Appl. Med. Bio.}, 14:\penalty0 39--69, 04 1997.
\newblock URL \url{10.1093/imammb14.1.39}.

\bibitem[Byrne and Preziosi(2003)]{Byrne_prezziozi_2003}
H.~M. Byrne and L.~Preziosi.
\newblock Modelling solid tumour growth using the theory of mixtures.
\newblock \emph{Math. Med. Bio.}, 20\penalty0 (4):\penalty0 341--366, 2003.
\newblock URL \url{10.1093/imammb/20.4.341}.

\bibitem[Breward et~al.(2001)Breward, Byrne, and Lewis]{IMA::breward_2001}
C.~J.~W. Breward, H.~M. Byrne, and C.~E. Lewis.
\newblock Modelling the interactions between tumour cells and a blood vessel in
  a microenvironment within a vascular tumour.
\newblock \emph{European J. Appl. Math.}, 12\penalty0 (5):\penalty0 529--556,
  2001.
\newblock URL \url{https://doi.org/10.1017/S095679250100448X}.

\bibitem[Byrne and Owen(2004)]{ByrneOwen}
H.~M. Byrne and M.~R. Owen.
\newblock A new interpretation of the keller-segel model based on multiphase
  modelling.
\newblock \emph{J. Math. Bio.}, 49\penalty0 (6):\penalty0 604--626, 2004.
\newblock URL \url{https://doi.org/10.1007/s00285-004-0276-4}.

\bibitem[Yao et~al.(2013)Yao, Liu, and Chen]{yao}
J.~Yao, W.~Liu, and Z.~Chen.
\newblock Numerical solution of a moving boundary problem of one-dimensional
  flow in semi-infinite long porous media with threshold pressure gradient.
\newblock \emph{Math. Problems Engg.}, 2013.
\newblock URL \url{http://dx.doi.org/10.1155/2013/384246}.

\bibitem[Eymard et~al.(2000)Eymard, Gallou\"{e}t, and Herbin]{eymard}
R.~Eymard, T.~Gallou\"{e}t, and R.~Herbin.
\newblock Finite volume methods.
\newblock In P.~G. Ciarlet and J.~L. Lions, editors, \emph{Solution of Equation
  in $\mathbf{R}^n$ (Part 3), Techniques of Scientific Computing (Part 3)},
  volume~7, pages 713--1018. Elsevier, Amsterdam, 2000.
\newblock URL \url{https://doi.org/10.1016/S1570-8659(00)07005-8}.

\bibitem[Brenner and Scott(2008)]{brenner2008mathematical}
S.~C. Brenner and L.~R. Scott.
\newblock \emph{The mathematical theory of finite element methods}.
\newblock Springer, 2008.
\newblock URL \url{https://doi.org/10.1007/978-0-387-75934-0}.

\bibitem[Ern and Guermond(2004)]{alexander}
A.~Ern and J.~Guermond.
\newblock \emph{Theory and Practice of Finite Elements}.
\newblock Applied mathematical sciences. Springer-Verlag New York, 2004.
\newblock URL \url{https://doi.org/10.1007/978-1-4757-4355-5}.

\bibitem[Droniou et~al.(2021)Droniou, Nataraj, and Remesan]{DNR19}
J.~Droniou, N.~Nataraj, and G.~C. Remesan.
\newblock Convergence analysis of a numerical scheme for a tumour growth model.
\newblock \emph{IMA J. Numer. Anal.}, drab016, 2021.
\newblock URL \url{https://doi.org/10.1093/imanum/drab016}.

\bibitem[Remesan(2019)]{remesan_1}
G.~C. Remesan.
\newblock Numerical solution of the two-phase tumour growth model with moving
  boundary.
\newblock \emph{ANZIAM J.}, 60:\penalty0 C1--C15, 2019.
\newblock URL \url{https://doi.org/10.21914/anziamj.v60i0.13936}.

\bibitem[Lewin et~al.(2020)Lewin, Maini, Moros, Enderling, and
  Byrne]{MMNP::Lewin:2020}
T.~D. Lewin, P.~K. Maini, E.~G. Moros, H.~Enderling, and H.~M. Byrne.
\newblock A three phase model to investigate the effects of dead material on
  the growth of avascular tumours.
\newblock \emph{Math. Model. Nat. Phenom.}, 15:\penalty0 22, 2020.
\newblock URL \url{https://doi.org/10.1051/mmnp/2019039}.

\bibitem[Park et~al.(2020)Park, Dharmasivam, and Richardson]{park}
K.~C. Park, M.~Dharmasivam, and D.~R. Richardson.
\newblock The role of extracellular proteases in tumor progression and the
  development of innovative metal ion chelators that inhibit their activity.
\newblock \emph{Int. J. Mol. Sci.}, 21\penalty0 (18), 2020.
\newblock URL \url{http://10.3390/ijms21186805}.

\bibitem[DeClerck et~al.(2004)DeClerck, Mercurio, Stack, Chapman, Zutter,
  Muschel, Raz, Matrisian, Sloane, Noel, Hendrix, Coussens, and
  Padarathsingh]{clerk}
Y.~A. DeClerck, A.~M. Mercurio, M.~S. Stack, H.~A. Chapman, M.~M. Zutter, R.~J.
  Muschel, A.~Raz, L.~M. Matrisian, B.~F. Sloane, A.~Noel, M.~J. Hendrix,
  L.~Coussens, and M.~Padarathsingh.
\newblock Proteases, extracellular matrix, and cancer: a workshop of the path b
  study section.
\newblock \emph{Am. J. Path.}, 4:\penalty0 1131--1139, 2004.
\newblock URL \url{http://10.1016/S0002-9440(10)63200-2}.

\bibitem[Franks et~al.(2003)Franks, Byrne, Mudhar, Underwood, and
  Lewis]{franks2003a}
S.~J. Franks, H.~M. Byrne, H.~S. Mudhar, J.~C. Underwood, and C.~E. Lewis.
\newblock Mathematical modelling of comedo ductal carcinoma in situ of the
  breast.
\newblock \emph{Math. Med. Biol.}, 20\penalty0 (3):\penalty0 277--308, 2003.
\newblock URL \url{https://doi.org/10.1093/imammb/20.3.277}.

\bibitem[Hubbard and Byrne(2013)]{IMA::hubbard}
M.~E. Hubbard and H.~M. Byrne.
\newblock Multiphase modelling of vascular tumour growth in two spatial
  dimensions.
\newblock \emph{J. Theoret. Biol.}, 316:\penalty0 70--89, 2013.
\newblock URL \url{10.1016/j.jtbi.2012.09.031}.

\bibitem[Webb et~al.(2007)Webb, Owen, Byrne, Murdoch, and Lewis]{web2007}
S.~D. Webb, M.~R. Owen, H.~M. Byrne, C.~Murdoch, and C.~E. Lewis.
\newblock Macrophage-based anti-cancer therapy: modelling different modes of
  tumour targeting.
\newblock \emph{Bull. Math. Biol.}, 69\penalty0 (5):\penalty0 1747--1776, 2007.
\newblock URL \url{https://doi.org/10.1007/s11538-006-9189-2}.

\bibitem[Owen et~al.(2004)Owen, Byrne, and Lewis]{owen2004}
M.~R. Owen, H.~M. Byrne, and C.~E. Lewis.
\newblock Mathematical modelling of the use of macrophages as vehicles for drug
  delivery to hypoxic tumour sites.
\newblock \emph{J. Theor. Biol.}, 226\penalty0 (4):\penalty0 377--391, 2004.
\newblock URL \url{https://doi.org/10.1016/j.jtbi.2003.09.004}.

\bibitem[Dyson et~al.(2016)Dyson, Green, Whiteley, and Byrne]{dyson2016}
R.~J. Dyson, J.~E. Green, J.~P. Whiteley, and H.~M. Byrne.
\newblock An investigation of the influence of extracellular matrix anisotropy
  and cell-matrix interactions on tissue architecture.
\newblock \emph{J. Math. Biol.}, 72\penalty0 (7):\penalty0 1775--1809, 2016.
\newblock URL \url{https://doi.org/doi:10.1007/s00285-015-0927-7}.

\bibitem[Esfahani et~al.(2020)Esfahani, Roudaia, Buhlaiga, Del~Rincon, Papneja,
  and Miller~Jr.]{esfahni2020}
K.~Esfahani, L.~Roudaia, N.~Buhlaiga, S.~V. Del~Rincon, N.~Papneja, and W.~H.
  Miller~Jr.
\newblock A review of cancer immunotherapy: from the past, to the present, to
  the future.
\newblock \emph{Curr. Oncol.}, 27\penalty0 (2):\penalty0 S87--S97, 2020.
\newblock URL \url{https://doi.org/10.3747/co.27.5223}.

\bibitem[Thariat et~al.(2013)Thariat, Hannoun-Levi, Sun~Myint, Vuong, and
  Gérard]{thariat2013}
J.~Thariat, J.~M. Hannoun-Levi, A.~Sun~Myint, T.~Vuong, and J.~P. Gérard.
\newblock Past, present, and future of radiotherapy for the benefit of
  patients.
\newblock \emph{Nat. Rev. Clin. Oncol}, 10\penalty0 (1):\penalty0 52--60, 2013.
\newblock URL \url{https://doi.org/10.1038/nrclinonc.2012.203}.

\bibitem[Chabner and Roberts~Jr.(2005)]{chabner2005}
B.~A. Chabner and T.~G Roberts~Jr.
\newblock Timeline: Chemotherapy and the war on cancer.
\newblock \emph{Nat. Rev. Cancer}, 5\penalty0 (1):\penalty0 65--72, 2005.
\newblock URL \url{https://doi.org/doi:10.1038/nrc1529}.

\bibitem[Giussani. et~al.(2019)Giussani., Triulzi, Sozzi, and
  Tagliabue]{guissani}
M.~Giussani., T.~Triulzi, G.~Sozzi, and E.~Tagliabue.
\newblock Tumor extracellular matrix remodeling: new perspectives as a
  circulating tool in the diagnosis and prognosis of solid tumors.
\newblock \emph{Cells}, 8\penalty0 (2), 2019.
\newblock URL \url{https://doi.org/10.3390/cells8020081}.

\bibitem[Alahmadi et~al.(2020)Alahmadi, Belet, Black, Cromer, Flegg, House,
  Jayasundara, Keith, McCaw, Moss, et~al.]{alahmadi2020influencing}
A.~Alahmadi, S.~Belet, A.~Black, D.~Cromer, J.~A. Flegg, T.~House,
  P.~Jayasundara, J.~M. Keith, J.~M. McCaw, R.~Moss, et~al.
\newblock Influencing public health policy with data-informed mathematical
  models of infectious diseases: Recent developments and new challenges.
\newblock \emph{Epidemics}, 32:\penalty0 100393, 2020.
\newblock URL \url{https://doi.org/10.1016/j.epidem.2020.100393}.

\end{thebibliography}

\appendix

\section{Long time behaviour if tumour length approaches domain boundary}
\label{sec:boxdomain}

The numerical results presented in the main text show that sustained tumour growth compresses the hydrogel. 
As the compressive stress that the hydrogel exerts on the tumour increases, eventually the tumour's growth rate tends to zero
(i.e., the tumour cannot grow indefinitely). We formalise this observation in Theorem~\ref{thm:steady_state}.

\begin{theorem}
	Let $(\alpha, u_{\alpha}, c, \ell)$ be a solution of Equations~\eqref{eqn:model} such that $u_{\alpha}(t,\ell(t)) \ge 0$ for every 
	$t \ge 0$. Suppose that there are positive constants  $\alpha_{\rm \ell}$ and $\alpha_{\rm u}$ such that 
	$0 < \alpha_{\ell} \le \alpha \le\alpha_{u} < 1$. 
	If $\vartheta \ge \nu$, then $\displaystyle \lim_{\ell(t) \rightarrow 1} \frac{\mathrm{d} \ell}{\mathrm{d}t}= 0$. 	
	\label{thm:steady_state}
\end{theorem}
\begin{proof}
	Multiply Equation~\eqref{eqn:cv_dless} by $u_{\alpha}(t,x)$ and apply integration by parts to obtain
	\begin{align}
	\int_{0}^{\ell(t)} \dfrac{ k \alpha}{1 - \alpha} u_{\alpha}^2\mathrm{d}x + \int_{0}^{\ell(t)} \mu \alpha \left(\dfrac{\partial u_{\alpha}}{\partial x}\right)^2 \,\mathrm{d}x = \int_{0}^{\ell(t)}  \mathscr{H}(\alpha)\dfrac{\partial u_{\alpha}}{\partial x}\,\mathrm{d}x + u_{\alpha}(t,\ell(t))  \sigma^{\mathrm{H}} . 
	\label{eqn:lt_1}
	\end{align}
	Since $u_{\alpha}(t,\ell(t)) \ge 0$ and  $\sigma^{\mathrm{H}} \le 0$, it follows that $u_{\alpha}(t,\ell(t))  \sigma^{\mathrm{H}} \le 0$. 
	Then Equation~\eqref{eqn:lt_1} supplies
	\begin{align} \label{eqn:ineq_1}
	\int_{0}^{\ell(t)} \mu \alpha \left(\dfrac{\partial u_{\alpha}}{\partial x}\right)^2 \,\mathrm{d}x \le \int_{0}^{\ell(t)}  \mathscr{H}(\alpha)\dfrac{\partial u_{\alpha}}{\partial x}\,\mathrm{d}x \le \sqrt{\ell(t)}\sup_{(0,\ell(t))}|\mathscr{H}(\alpha)| \left(  \int_{0}^{\ell(t)}  \left(\dfrac{\partial u_{\alpha}}{\partial x} \right)^2\,\mathrm{d}x \right)^{1/2}
	\end{align}
	where we have applied the Cauchy-Schwartz inequality. 
	Since $0 < \alpha_{\ell} \le \alpha \le \alpha_{u} < 1$, Equation~\eqref{eqn:ineq_1} yields 
	\begin{align}
	\mu \alpha_\ell \left(\int_{0}^{\ell(t)} \left(\dfrac{\partial u_\alpha}{\partial x}\right)^2 \,\mathrm{d}x \right)^{1/2} \le \sqrt{\ell(t)} \dfrac{\alpha_{u}|\alpha_{u}  - \alast |}{|1 - \alpha_{u}|^2}. 
	\label{eqn:lt_2}
	\end{align}
	Since the left-hand side of Equation~\eqref{eqn:lt_1} is non-negative, we obtain 
	$$u_{\alpha}(t,\ell(t)) \le \frac{1}{ |\sigma^{\mathrm{H}}|}\int_{0}^{\ell(t)}  \mathscr{H}(\alpha)\frac{\partial u_{\alpha}}{\partial x}\,\mathrm{d}x. $$ 
	Then, application of Cauchy-Schwartz inequality to $\int_{0}^{\ell(t)}  \mathscr{H}(\alpha)\frac{\partial u_{\alpha}}{\partial x}\,\mathrm{d}x$, together with Equation~\eqref{eqn:lt_2}, leads to 
	\begin{align}
	u_{\alpha}(t,\ell(t)) \le \dfrac{\ell(t)}{ \mu \alpha_{\ell}|\sigma^{\mathrm{H}}_{\vert \ell(t)}|} \left( \dfrac{\alpha_{u}|\alpha_{u}  - \alast |}{|1 - \alpha_{u}|^2} \right)^2.
	\end{align}
	From Equation~\eqref{eqn:gel}, we have that $|\sigma^{\mathrm{H}}| \rightarrow \infty$ as $\ell(t) \rightarrow 1$. 
	Therefore, it follows that
	$$\lim_{\ell(t) \rightarrow 1} \frac{\mathrm{d} \ell}{\mathrm{d}t} = \lim_{\ell(t) \rightarrow 1} u_{\alpha}(t,\ell(t)) = 0,$$ 
	as required. 
	\end{proof}

\begin{remark}
	Theorem~\ref{thm:steady_state} does not explicitly describe the time evolution of the tumour length. 
	Rather, it asserts that the cell velocity approaches zero as the tumour approaches the domain boundary (i.e., as $\ell(t) \rightarrow 1$). 
	Thus, the tumour length undergoes sigmoidal growth, as in Figure~\ref{fig:nu_const}, wherein $\ell(t)$ increases monotonically at early times and then plateaus at a stationary value. 
	This theoretical result is consistent with the experimental observations in~\cite{Helmlinger1997}. 
\end{remark}

\end{document}